\documentclass[12pt]{article}

\usepackage{preamble_paper}

\title{Mechanism Design with Information Leakage\thanks{We thank Sarah Auster, Laura Doval, Piotr Dworczak, Paul Klemperer, Elliot Lipnowski, Nick Netzer, Axel Niemeyer, Axel Ockenfels, Armin Schmutzler, Ilya Segal, and seminar participants at U Zürich, U St.~Gallen, Conference on Mechanism and Institution Design, Bonn HIM Workshop, and 37th Stony Brook International Conference on Game Theory for their comments. We thank Martin Valdevenito Belmar and Kiril Lavrov for research assistance. This project has received funding from the Swiss National Science Foundation.
}}
\author{{\small
    \begin{tabular}{ccc}
      {\bf Samuel Häfner} & {\bf Marek Pycia} & {\bf Haoyuan Zeng}  \\
      University of St.Gallen & University of Zurich & Yale University \\
      \href{mailto:samuel.haefner@unisg.ch}{samuel.haefner@unisg.ch} & \href{mailto:marek.pycia@econ.uzh.ch}{marek.pycia@econ.uzh.ch} & \href{mailto:haoyuan.zeng@yale.edu}{haoyuan.zeng@yale.edu} \\
      \rule{1in}{0pt}& \rule{1in}{0pt} & \rule{1in}{0pt}\\
\end{tabular}}}
\date{October 2026}

\begin{document}

\maketitle

\begin{abstract}
  We study the design of mechanisms---e.g., auctions---when the designer does not control information flows between mechanism participants.
  A mechanism equilibrium is leakage-proof if no player conditions their actions on leaked information; a property distinct from both ex-post and posterior incentive compatibility.
  Only leakage-proof mechanisms can implement social choice functions in environments with leakage; efficient auctions and revenue-maximizing ones need to be leakage-proof.
  Second-price and ascending auctions are leakage-proof, while first-price auctions need not be. Appropriately designed Dutch auctions are leakage-proof and attain optimal revenue even with irregular value distributions. In multi-unit environments, leakage-proof mechanisms can attain the optimal Bayesian revenue even when dominant-strategy mechanisms cannot.

  \medskip

  \medskip
  \JELcodes{D44, D47, D82, D83}
\end{abstract}

\newpage
\section{Introduction}
Standard mechanism design focuses on environments in which the mechanism designer has complete control over the extensive-form game the mechanism participants play, including complete control over who observes the moves they make. This complete control assumption underlies the revelation principle \citep{Myerson-OptimalAuctionDesign-1981} and many other insights in mechanism design. In many modern market environments, however, the designer cannot fully control what players observe. In this paper, we study mechanism design in such environments.

For example, many online auctions are prone to so-called eavesdropping attacks, in which some participants can listen to the network for upcoming bids and use this information to their advantage when submitting their own bids; in the presence of eavesdropping, the designer does not have complete control over what the bidders observe.\footnote{For institutional details on eavesdropping attacks, see, e.g., \citet*{Franklin-Reiter-DesignImplementationSecure-1996}.} In financial markets, participants with fast access to trading venues can be the first to assess the limit order book and execute their trade orders more quickly than competitors, thereby capturing transient arbitrage opportunities. These market participants have more flexibility over what they observe than the market designer may want to grant them.\footnote{The speed of access is so vital that some large traders place their servers as close to the exchange as possible and connect the two via dedicated optical fiber connections, resulting in what is sometimes referred to as a high-frequency trading arms race \citep*{Budish-Cramton-Shim-HighFrequencyTradingArms-2015}.} In peer-to-peer networks, such as blockchains, faster participants also observe slower participants' actions. Exploiting this information---known as front-running---frequently occurs on so-called decentralized exchanges and is regarded as one of the most pressing problems of blockchain technology.\footnote{See, e.g., \citet*{Eskandari-Moosavi-Clark-SoKTransparentDishonesty-2020}. Blockchains are ledgers that record transactions. New transactions are added to blockchains in batches, called blocks. The production of new blocks happens in discrete time intervals. Between block production, recent transactions are stored in the network. By design, the network is open, and participants can inspect upcoming transactions.}

To see the implications of information leakage (or eavesdropping), consider the sale of an object to one of two bidders who have private values drawn independently from the same distribution. In the absence of information leakage, the first-price auction with no reserve is efficient, while the first-price auction with an optimal reserve is revenue maximizing under standard regularity assumptions. Both properties can fail in the presence of information leakage. Indeed, suppose that the bid of the first bidder is seen by the second bidder before the second bidder submits her bid. The second bidder may then condition her bid on the bid of the first bidder. If the first bidder bids below her value, the second bidder may win the auction even when her value is lower than the first bidder's, undermining efficiency. Similarly, information leakage can prevent the auction from allocating the object to the bidder with the highest Myersonian virtual value, undermining revenue maximization.

We study a designer who can choose any finite extensive-form game with perfect recall but does not fully control what history of the game the players see. Even when two or more players are to move simultaneously in the game chosen by the designer, some of them (``faster'' players) might see the moves made by others (``slower'' players) before deciding on their own move.\footnote{For expositional reasons, we focus on games that allow simultaneous moves but otherwise have perfect information. All our insights remain true for general imperfect-information games. Not only do all of our arguments extend to the general case, but also our results directly imply the analogous results for the general case. The reason is simple: we are interested in mechanisms that perform well in the presence of leakage,  including leakage that reveals to every player the history of past play.} The information about players' moves might hence leak to other players, and the designer cannot prevent it. We allow players to be uncertain about who observes whom and to have private information about the realized leakage structure. The designer cannot condition the rules of the game on this structure.

In this setting, we first address the question of what extensive-form games can implement a fixed interim-payoff schedule, possibly subject to additional requirements on the implemented outcomes. We show that a game implements such a schedule under every possible leakage structure if and only if, after removing histories that are never reached in equilibrium, it admits an equilibrium in which all players pursue strategies that are independent of the leaked information and that implement the same schedule and satisfy the outcome requirements. We call such equilibria---in which leaked information is effectively ignored by the players---leakage-proof. Leakage-proofness is hence not only sufficient but also necessary for interim implementation in environments in which the designer does not control information flows between the players.

We show that leakage-proofness is distinct from dominant strategies as well as from ex-post and posterior incentive compatibility.\footnote{We also show that leakage-proofness is equivalent to a weakened analogue of posterior implementability that we introduce and call sequential posterior implementability.} In particular, the class of leakage-proof mechanisms is broader than the class of dominant-strategy mechanisms, and, e.g., they can be used to maximize revenue in some economically relevant environments in which dominant-strategy mechanisms cannot do so.

The possibility of leakage also highlights the importance of an auction's extensive form. Without leakage, first-price and descending (Dutch) auctions are strategically equivalent. Under leakage, however, they can perform differently. A first-price auction allows bidders to respond to their opponents' bids, which is detrimental to its revenue and efficiency performance,  whereas a Dutch auction can be designed so that only one bidder acts at each stage and it is then leakage-proof.

We then study revenue-maximizing auctions and efficient auctions. We show that, for any Bayesian incentive-compatible allocation rule, revenue maximization under leakage identifies a unique interim-payoff schedule: each bidder must receive the lowest rents consistent with incentive compatibility and individual rationality. Consequently, an auction that implements the allocation rule and maximizes revenue conditional on every leakage structure can be transformed into a leakage-proof auction that implements the same allocation and attains the same revenue.
This result allows us to apply our general implementation theorem to revenue maximization even though fixing the allocation rule alone does not determine bidders' payoffs.

In single-unit auctions, leakage-proof mechanisms achieves the same maximum expected revenue as achievable in the standard model without leakage.
Second-price and Dutch auctions are both leakage-proof, and when bidders' values are identically distributed and regular, they maximize revenue when properly designed (with optimal reserve price). When value distributions are irregular, second-price auctions cannot necessarily be designed to maximize revenue, but appropriately designed Dutch auctions can (a point made by \citet*{Jeong-Pycia-FirstPricePrincipleFailure-2025}).
In single-unit environments, there are other dominant-strategy mechanisms that maximize revenue; indeed, in the absence of leakage, dominant-strategy and Bayesian mechanisms achieve the same revenue \citep*{Gershkov-Goeree-Kushnir-Moldovanu-Shi-EquivalenceBayesianDominant-2013}, an equivalence that straightforwardly extends to dominant-strategy and leakage-proof mechanisms.

In multi-unit environments, leakage-proof mechanisms can strictly revenue dominate dominant-strategy mechanisms. When multiple identical units are sold, and bidders may demand more than one unit, Bayesian and dominant-strategy incentive compatibility are no longer equivalent in terms of optimal revenue. We show that, while leakage-proof mechanisms do not always attain the optimal Bayesian revenue achievable without leakage, they can attain this bound in environments in which dominant-strategy mechanisms cannot. Also, when bidders' marginal values and virtual marginal values are appropriately ordered, leakage-proof mechanisms can achieve the Bayesian revenue bound.

Leakage-proof mechanisms can also achieve efficient allocation rules in both single-unit and mulit-unit settings. As noted above, leakage-proofness is necessary to achieve effciency in environments with leakage, and, e.g., first-price auctions failing to be leakage-proof are also not efficient.

Finally, we introduce the concept of sequential posterior implementation that requires that, at every information set, the prescribed continuation strategies remain optimal even when a player observes all concurrent and past actions of the other players. We show that sequential posterior implementation and leakage-proofness are equivalent. We also establish the relationship between leakage-proofness and ex-post incentive compatibility. Pure-strategy ex-post incentive-compatible equilibria are always leakage-proof, but mixed-strategy ex-post equilibria are not necessarily so.
Conversely, leakage-proof equilibria need not be ex-post incentive compatible.

While in the results discussed above we study standard common prior environments, in the paper we also discuss leakage-proofness in environments in which different players can have different priors about the leakage process, and hence the beliefs of at least some players are misspecified.
Our implementation results extend to such environments, but the necessity of leakage-proofness for revenue maximization does not. The designer who knows the leakage environment can achieve higher revenue than the no-leakage Bayesian revenue bound.

\subsection{Related Literature}

While we believe we are the first to systematically study information leakage in a mechanism design framework, information flows have been studied in games and applied economic contexts.
In particular, \cite{Solan-Yariv-GamesEspionage-2004} consider two-player normal-form games, in which one player can noisily observe the strategy of the other player at some cost, and fully characterize the distributions over the players' payoffs that can obtain in equilibrium.
\citet{Penta-Zuazo-Garin-RationalizabilityObservabilityCommon-2022} characterize the predictions of rationality and common belief in rationality that do not depend on players' infinite order beliefs over whether their actions are observable to their opponents.\footnote{For an analysis of repeated games with mediation, in which mediators' recommendations ``leak'' after each period, see \citet{Ewerhart-Zeng-MediatedSubgamePerfect-2026}.
  Auctions are sometimes designed with an explicit leakage component in order to privilege a specific bidder, who has the right-of-first refusal.
Such rights are commonly seen in the landlord-tenant and competitive procurement settings; see \cite{Riley-Samuelson-OptimalAuctions-1981}, \cite*{Burguet-Perry-PreferredSuppliersAuction-2009}, \cite*{Bikhchandani-Lippman-Ryan-RightofFirstRefusal-2005},  \cite*{Choi-RentExtractionTheory-2009}, and \cite{Doran-AsymmetricRightFirst-2018}. }

A solution concept aiming to capture information leakage in static games is posterior equilibrium, introduced by \citet{Green-Laffont-PosteriorImplementabilityTwoPerson-1987} and studied, e.g., in \citet{Lopomo-OptimalityRobustnessEnglish-2001}, \citet*{Jehiel-Meyer-ter-Vehn-Moldovanu-Zame-PosteriorImplementationVs-2007}, and \citet{Niemeyer-PosteriorImplementabilityCollective-2022}. Observing that, beyond static games, posterior equilibrium is stronger than a leakage-proofness requirement and hence unnecessarily restrictive, we introduce its dynamic analog --- which we call sequential posterior equilibrium---and we prove that the sequential posterior solution concept is equivalent to leakage proofness. We thus provide a formal motivation to study posterior and sequential posterior equilibria.

Ex-post equilibrium \citep*{Hurwicz-InformationallyDecentralizedSystems-1972,Dasgupta-Hammond-Maskin-ImplementationSocialChoice-1979} may be motivated by the concern about leakage of information about players' types, as opposed to actions.\footnote{The concern over the leakage of players' types spun a rich literature on privacy-preserving mechanisms; see, e.g.,
  \cite*{Izmalkov-Micali-Lepinski-RationalSecureComputation-2005}, \cite*{Brandt-Sandholm-ExistenceUnconditionallyPrivacyPreserving-2008}, \cite{Milgrom-Segal-ClockAuctionsRadio-2020}, \cite{Haupt-Hitzig-ContextuallyPrivateMechanisms-2024},
  \cite{Gao-Tang-PrivateWhomMinimal-2026}.
}
As we show, these two concerns are distinct (see Section \ref{sec:expost}).

By studying leakage, we contribute to the recent wave of literature that differentiate between auction mechanism on the ground of their robustness to attacks such as (the lack of) credibility \citep*{Akbarpour-Li-CredibleAuctionsTrilemma-2020, Banchio-Skrzypacz-Yang-DynamicThreatsCredible-2025}, auditability \citep*{Woodward-SelfAuditableAuctions-2020}, or shill-bidding \citep*[e.g.,][]{Komo-Kominers-Roughgarden-ShillProofAuctions-2024,Zeng-IdentityCompatibleAuctions-2025}.
While this literature focuses on the seller's inability to fully commit, we study the seller's inability to control information flows.

Our analysis of leakage is relevant for many applied problems.
In financial markets, traders often use faster access to information about upcoming trades to exploit arbitrage opportunities \cite*[see, e.g.,][for extensive evidence]{Budish-Cramton-Shim-HighFrequencyTradingArms-2015, Baldauf-Mollner-FastTradersMake-2022}. \cite*{Budish-Cramton-Shim-HighFrequencyTradingArms-2015} propose batch auctions to limit the adverse effects of such front-running, while we focus on establishing the link between robust implementation and leakage-freeness.
\footnote{Information flows between participants matter also in treasury auctions and over-the-counter markets; for the discussion of evidence see, e.g., \cite{Hortacsu-Sareen-OrderFlowFormation-2005}, \cite{Hortacsu-Kastl-ValuingDealersInformational-2012}, and \cite*{Garratt-Lee-Martin-Townsend-WhoSeesTrades-2019}.}

Leakage is also a major issue in online markets and there is an extensive computer science literature on cryptographic approaches to make auctions robust to information leakage \citep*{Kudo-SecureElectronicSealedBid-1998,Abe-Suzuki-1stPriceAuction-2002,Parkes-Rabin-Shieber-Thorpe-PracticalSecrecypreservingVerifiably-2008}.\footnote{Leakage of information about actions is related to sniping that is submitting one's bid as close as possible to the end-time of the auction \citep*{Roth-Ockenfels-LastMinuteBiddingRules-2002}. Our analysis of leakage-proofness might be viewed as addressing the inefficiencies associated with sniping.}
Cryptographic methods that allow to shield bids during the auction and to verify them (or at least their ranking) after the auction have received renewed interest in decentralized blockchain settings, lately \citep*{Blass-Kerschbaum-StrainSecureAuction-2018,Blass-Kerschbaum-BOREALISBuildingBlock-2020,Galal-Youssef-VerifiableSealedBidAuction-2019}. We view our analysis as complementary to these approaches, because we focus on incentives, rather than cryptography, to prevent bidders from using leaked information.

Another previously studied solution to leakage is candle auctions, in which the acceptance of bids is uncertain. \cite*{Hafner-Stewart-FrontRunningAscendingAuctions-2021} show that a candle auction may implement the efficient and optimal outcome despite leakage, while \cite*{Gehrlein-Hafner-Oechssler-CandleAuctionField-2025} show that candle auctions admit an approximate ex-post equilibrium in quickly escalating bids.

Our result that leakage-proof implementation is more flexible than dominant-strategy implementation is complementary to the literature studying when interim allocations of Bayesian Incentive Compatible mechanism can be implemented in dominant strategies \citep*[see, e.g.,][]{Manelli-Vincent-BayesianDominantStrategyImplementation-2010,Gershkov-Goeree-Kushnir-Moldovanu-Shi-EquivalenceBayesianDominant-2013}.

We contribute to the study of canonical auction formats by showing that not only dominant-strategy mechanisms but also properly designed Dutch auctions are leakage-proof. Our result here is complementary to \citet{Jeong-Pycia-FirstPricePrincipleFailure-2025}, who show that Dutch auctions are sufficient to maximize a wide variety of objectives.

The extension of our results to environments with misspecified beliefs (heterogeneous priors), which we provide in Section \ref{misspecification}, is complementary to the large literature on  misspecified beliefs in mechanisms and  markets  \citep*{Ledyard-IncentiveCompatibilityIncomplete-1978,Bergemann-Morris-RobustMechanismDesign-2005,
  Chassang-CalibratedIncentiveContracts-2013,
  Carroll-RobustnessLinearContracts-2015,
  Madarasz-Prat-SellersMisspecifiedModels-2017,
  Li-ObviouslyStrategyProofMechanisms-2017,
  Heidhues-Koszegi-Strack-UnrealisticExpectationsMisguided-2018,
  Borgers-Li-StrategicallySimpleMechanisms-2019,
Pycia-Troyan-TheorySimplicityGames-2023}.
While the main thrust of this literature is that robustness to misspecification requires the mechanism to be simple, we show that efficiency can be achieved in a way that is robust to misspecification and requires the mechanism to be leakage-proof,
while knowing the details of misspecified beliefs is crucial for revenue-maximization in the presence of leakage.

\section{Model}

\subsection{Environment}

We develop the results of Section 4 for general environments in which we allow value interdependencies, type correlations, and non-quasilinear utilities. We later (in Section 5) apply these general results to environments with independent types and quasilinear utilities.

\paragraph{Players, Outcomes, and Payoffs.}
In our baseline model, the set of players is \(N=\{1,\dots,n\}\).
Player \(i\)'s payoff type is \(\theta_i\in\Theta_i\), where \(\Theta_i\) is finite. We write \(\Theta:=\prod_{i\in N}\Theta_i\) and \(\theta=(\theta_i)_{i\in N}\in\Theta\) for a payoff-type profile. The players share a common prior \(f\in\Delta(\Theta)\), which has full support on \(\Theta\). The finite set of outcomes is \(X\), and player \(i\)'s utility function is \(u_i:X\times\Theta\to\mathbb R\). The payoff structure is common knowledge.

\paragraph{Extensive Forms.}
We are interested in the design of extensive-form games in environments in which the designer may be unable to control agents' ability to see other agents' moves, which we refer to as information leakage.
Static (simultaneous-move) games are a special case.
For simplicity, we restrict attention to finite extensive-form games with perfect recall that allow simultaneous moves but otherwise have perfect information. This restriction simplifies the exposition but is not essential: our arguments extend to general imperfect-information games.
As we study mechanisms that remain robust to information leakage, our results directly imply analogous conclusions for designers who choose from among all finite imperfect information games in environments with perfect recall in which leakage may reveal some or all of the otherwise unobserved history.

Formally, a multistage game with observed actions is
\[
  G
  =
  \left(
    H,\subset,P,(A_i)_{i\in N},g
  \right).
\]
The primitives are as follows.

\begin{enumerate}
  \item \(H\) is a finite set of \textbf{public histories}. The public history at the beginning of stage \(1\) is the empty history \(h_\emptyset\), which we also denote by \(h^1\). At the beginning of stage \(t\geq2\), the public history is \(h^t=(a^1,\dots,a^{t-1})\in H\), where \(a^s=(a_1^s,\dots,a_n^s)\) records the actions taken at stage \(s\). We use \(a_i^s=\emptyset\) when player \(i\) is inactive at stage \(s\). If \(h^{t'}\) is a strict prefix of \(h^t\), we write \(h^{t'}\subset h^t\).

  \item The set of terminal histories is denoted by \(Z\subseteq H\).

  \item The player function \(P:H\setminus Z\to2^N\) assigns to each nonterminal public history \(h^t\) the set \(P(h^t)\) of players who are active at stage \(t\).

  \item For every \(h^t\in H\setminus Z\) and \(i\in N\), \(A_i(h^t)\) is player \(i\)'s action set at \(h^t\). We require \(A_i(h^t)=\{\emptyset\}\) if \(i\notin P(h^t)\), \(|A_i(h^t)|\geq2\) if \(i\in P(h^t)\), and
    \[
      \{a^t:(h^t,a^t)\in H\}
      =
      \prod_{i\in N}A_i(h^t).
    \]

  \item The outcome rule \(g:Z\to X\) assigns an outcome to each terminal history.
\end{enumerate}

Thus, at the beginning of each stage \(t\), the current public history \(h^t\) is observed by all players. The players in \(P(h^t)\) then take their stage-\(t\) actions. The resulting action profile \(a^t\) determines the next public history \(h^{t+1}=(h^t,a^t)\). In particular, the simplifying restriction we made implies that, absent information leakage, actions taken within the same stage are simultaneous, while all past actions are publicly observed.

\paragraph{Leakage Orders.}

We allow for noncontractible information leakage about actions taken simultaneously within a stage. Because such leakage is not verifiable, the designer cannot condition the rules of the game on who observes whom.

A \textbf{leakage order} \(\omega\in\mathcal O_N\) is represented by a preorder \(\leq_\omega\) on \(N\), that is, a reflexive and transitive binary relation. We write \(i\sim_\omega j\) if \(i\leq_\omega j\) and \(j\leq_\omega i\), and \(i<_\omega j\) if \(i\leq_\omega j\) but \(j\not\leq_\omega i\).

If \(i<_\omega j\), player \(j\) observes player \(i\)'s concurrent action whenever both players are active at the same stage, whereas player \(i\) does not observe player \(j\)'s. We say that \(j\) is more \textbf{informed} than \(i\), or equivalently that \(i\) is less informed than \(j\). Thus, information leaks from less informed players to more informed players. If \(i\sim_\omega j\), we say that \(i\) and \(j\) are equally informed; in this case, neither observes the other's concurrent action. We let \(\mathcal O_N\) denote the set of leakage orders on \(N\).\footnote{For simplicity, we omit explicitly modeling Nature's (chance) moves in the formal description. All our results extend to finite games with Nature's moves whose realizations are publicly observed and whose distributions may depend on the public history but not on payoff types or leakage orders. Restricting Nature to be the sole mover at each information set at which it moves is without loss for the designer. This allows mechanisms to randomize over terminal outcomes.
}

\paragraph{Leakage Type Space.}

We allow players to be uncertain about the realized leakage order. Since a leakage order completely determines who observes whose concurrent actions, we take the leakage order itself as the state of leakage.\footnote{More generally, one may introduce a finite state space \(\Omega\), a prior on \(\Omega\), a map \(\rho:\Omega\to\mathcal O_N\) determining the realized leakage order, and leakage-type maps \(\ell_i:\Omega\to\Lambda_i\). The formulation in the text corresponds to taking the leakage order itself as the state. The more general formulation allows distinct underlying states to induce the same leakage order while generating different private information for the players.}

A finite leakage type space is
\[
  \mathcal L
  =
  \left(
    (\Lambda_i)_{i\in N},
    \nu,
    (\ell_i)_{i\in N}
  \right),
\]
where \(\nu\in\Delta(\mathcal O_N)\) is the common prior over leakage orders and, for each \(i\in N\), \(\ell_i:\mathcal O_N\to\Lambda_i\) is player \(i\)'s leakage-type map. If the realized leakage order is \(\omega\), player \(i\) privately observes the leakage type \(\lambda_i=\ell_i(\omega)\). Without loss of generality, we take \(\Lambda_i=\ell_i(\operatorname{supp}\nu)\).

We require that a player's leakage type determines which players' concurrent actions she observes. Thus, for every \(i\in N\) and \(\omega,\omega'\in\operatorname{supp}\nu\),
\[
  \ell_i(\omega)=\ell_i(\omega')
  \quad\Longrightarrow\quad
  \{j\in N:j<_\omega i\}
  =
  \{j\in N:j<_{\omega'}i\}.
\]
Accordingly, for each \(\lambda_i\in\Lambda_i\), define
\[
  O_i(\lambda_i)
  :=
  \{j\in N:j<_\omega i\},
  \qquad
  \text{for any \(\omega\in\operatorname{supp}\nu\) such that
  \(\ell_i(\omega)=\lambda_i\)}.
\]
The preceding requirement ensures that \(O_i(\lambda_i)\) is well defined. Thus, a leakage type determines player \(i\)'s own pattern of leakage but need not reveal the entire leakage order. In particular, two distinct leakage orders may induce the same leakage type for player \(i\) while differing in the information available to other players.

For every \(\lambda_i\in\Lambda_i\), player \(i\)'s posterior over leakage orders is
\[
  \nu_i(\omega\mid\lambda_i)
  =
  \frac{
    \nu(\omega)\mathbf 1_{\{\ell_i(\omega)=\lambda_i\}}
  }{
    \displaystyle
    \sum_{\omega'\in\mathcal O_N:\,\ell_i(\omega')=\lambda_i}
    \nu(\omega')
  }.
\]
Thus, a leakage type is a private signal of the realized leakage order rather than the order itself.

\subsection{Equilibrium}

Fix a leakage type space \(\mathcal L\). The corresponding
\textbf{leakage environment} is \(\Gamma=(G,\mathcal L)\).

At stage \(t\), following public history \(h^t\), an active player \(i\in P(h^t)\) with leakage type \(\lambda_i\) observes the concurrent actions of the active players in \(O_i(\lambda_i)\). Her \textbf{private history} is
\[
  h_i^t
  =
  \left(
    h^t,
    a^t_{O_i(\lambda_i)\cap P(h^t)}
  \right).
\]
Thus, \(h_i^t\) consists of the public history together with the concurrent actions that have leaked to player \(i\). If \(O_i(\lambda_i)\cap P(h^t)=\varnothing\), we identify \(h_i^t\) with \(h^t\). Let \(H_i(\lambda_i)\) denote the set of private histories at which player \(i\) may be called to move when her leakage type is \(\lambda_i\).

For each \(\lambda_i\in\Lambda_i\), let \(\Sigma_i(\lambda_i)\) denote
the set of behavioral strategies that assign to each private history
\(h_i^t\in H_i(\lambda_i)\) a mixed action in
\(\Delta(A_i(h^t))\), where \(h^t\) is the public component of
\(h_i^t\). A strategy for player \(i\) is a map \(S_i\) such that
\(S_i(\theta_i,\lambda_i)\in\Sigma_i(\lambda_i)\) for every
\((\theta_i,\lambda_i)\in\Theta_i\times\Lambda_i\). Thus,
\[
  S_i(\theta_i,\lambda_i)(h_i^t)
  \in
  \Delta(A_i(h^t)).
\]
We write \(S=(S_i)_{i\in N}\) for a strategy profile.

A belief system specifies, for each player \(i\), type
\((\theta_i,\lambda_i)\), and private history
\(h_i^t\in H_i(\lambda_i)\), a posterior
\[
  \mu_i(\,\cdot\mid\theta_i,\lambda_i,h_i^t)
  \in
  \Delta(\Theta_{-i}\times\mathcal O_N).
\]
Payoff types and leakage orders are independently distributed according
to \(f\otimes\nu\). Hence, before play begins, player \(i\)'s posterior
conditional on \((\theta_i,\lambda_i)\) is
\(f(\theta_{-i}\mid\theta_i)\nu_i(\omega\mid\lambda_i)\). We write
\(\mu=(\mu_i)_{i\in N}\) for a belief profile.

Given \(S\), \(\theta\), and \(\omega\), player \(i\)'s realized leakage
type is \(\ell_i(\omega)\), so she follows
\(S_i(\theta_i,\ell_i(\omega))\). We write
\(\Pr(h_i^t\mid\theta,\omega;S)\) for the probability that player \(i\)
reaches private history \(h_i^t\), and
\(\Pr(z\mid\theta,\omega;S)\) for the probability of terminal history
\(z\in Z\).

Conditional on player \(i\)'s own types, define
\[
  \Pr(h_i^t\mid\theta_i,\lambda_i;S)
  :=
  \sum_{\theta_{-i}\in\Theta_{-i}}
  \sum_{\omega\in\mathcal O_N}
  f(\theta_{-i}\mid\theta_i)
  \nu_i(\omega\mid\lambda_i)
  \Pr\bigl(
    h_i^t
    \mid
    (\theta_i,\theta_{-i}),\omega;
    S
  \bigr).
\]

For an alternative strategy \(\widetilde S_i\in\Sigma_i(\lambda_i)\),
we write
\[
  \Pr\left(
    z
    \mid
    \theta,\omega,h_i^t;
    \widetilde S_i,S_{-i}
  \right)
\]
for the probability of terminal history \(z\), conditional on reaching
\(h_i^t\), when player \(i\) follows \(\widetilde S_i\) from \(h_i^t\)
onward and the other players follow \(S_{-i}\).

Our equilibrium concept is perfect Bayesian equilibrium (PBE)
\citep{Fudenberg-Tirole-PerfectBayesianEquilibrium-1991}.

\begin{definition}
  \label{def:eqm}
  An assessment \((S,\mu)\) is a \textbf{perfect Bayesian
  equilibrium} of the leakage environment \(\Gamma=(G,\mathcal L)\) if the
  following conditions hold.

  \begin{enumerate}
    \item \textbf{Sequential rationality.}
      For every \(i\in N\), \(\theta_i\in\Theta_i\),
      \(\lambda_i\in\Lambda_i\), and \(h_i^t\in H_i(\lambda_i)\),
      \begin{align*}
        S_i(\theta_i,\lambda_i)
        \in
        \arg\max_{\widetilde S_i\in\Sigma_i(\lambda_i)}
        &\sum_{\theta_{-i}\in\Theta_{-i}}
        \sum_{\omega\in\mathcal O_N}
        \mu_i(
          \theta_{-i},\omega
          \mid
          \theta_i,\lambda_i,h_i^t
        )
        \\
        &\quad\times
        \sum_{z\in Z}
        u_i\bigl(g(z),(\theta_i,\theta_{-i})\bigr)
        \Pr\left(
          z
          \mid
          (\theta_i,\theta_{-i}),
          \omega,
          h_i^t;
          \widetilde S_i,S_{-i}
        \right).
      \end{align*}

    \item \textbf{Bayesian updating.}
      At every private history \(h_i^t\in H_i(\lambda_i)\) such that
      \(\Pr(h_i^t\mid\theta_i,\lambda_i;S)>0\), beliefs satisfy
      \begin{align*}
        &\mu_i(
          \theta_{-i},\omega
          \mid
          \theta_i,\lambda_i,h_i^t
        )
        \\
        &\quad=
        \frac{
          f(\theta_{-i}\mid\theta_i)
          \nu_i(\omega\mid\lambda_i)
          \Pr\bigl(
            h_i^t
            \mid
            (\theta_i,\theta_{-i}),\omega;
            S
          \bigr)
        }{
          \displaystyle
          \sum_{\hat\theta_{-i}\in\Theta_{-i}}
          \sum_{\omega'\in\mathcal O_N}
          f(\hat\theta_{-i}\mid\theta_i)
          \nu_i(\omega'\mid\lambda_i)
          \Pr\bigl(
            h_i^t
            \mid
            (\theta_i,\hat\theta_{-i}),\omega';
            S
          \bigr)
        }.
      \end{align*}
      At a zero-probability private history, beliefs may be arbitrary over
      \(\Theta_{-i}\times
      \{\omega\in\mathcal O_N:\ell_i(\omega)=\lambda_i\}\).
  \end{enumerate}
\end{definition}

When we compare equilibria across leakage type spaces, we write
\((S^{\mathcal L},\mu^{\mathcal L})\) for a PBE of
\((G,\mathcal L)\).

\subsection{Pruning and Default Mechanisms}

Pruning allows us to dispense with histories that are irrelevant for a
given candidate strategy profile. Because relevance depends on the
candidate behavior, the pruning operation is defined relative to that
profile, before imposing the equilibrium requirement.

Fix a game \(G\) and a candidate strategy profile
\(S^G=(S_i^G)_{i\in N}\), where
\(S_i^G(\theta_i)(h^t)\in\Delta(A_i(h^t))\). For each
\(\theta\in\Theta\), let \(H^G(\theta)\) denote the set of public
histories reached with positive probability when players follow
\(S^G(\theta)\).

\begin{definition}
  \label{def:pruning}
  Fix a game \(G\) and a candidate strategy profile \(S^G\). The
  \textbf{game pruned relative to \(S^G\)}, denoted by \(G^p\), is the
  restriction of \(G\) to
  \[
    H^p
    :=
    \bigcup_{\theta\in\Theta}H^G(\theta).
  \]
  The player function, feasible actions, and outcome rule of \(G^p\) are
  obtained by restricting those of \(G\) to the histories and branches
  contained in \(H^p\).
\end{definition}

Pruning removes exactly those public histories that cannot arise under
the candidate strategy profile \(S^G\) for any payoff-type profile. A
public history in \(H^p\) may nevertheless be off path for a particular
payoff type if it is reached under \(S^G\) for some other payoff-type
profile. Thus, pruning retains every public history that can arise under
the candidate strategy profile while dispensing with histories that are
irrelevant for all payoff-type profiles.

Let \(S^{G^p}\) denote the restriction of \(S^G\) to the pruned game
\(G^p\). If \(S^{G^p}\) is a PBE strategy profile of \(G^p\), we call
\(S^{G^p}\) the \textbf{default strategy profile}, \(G^p\) the
\textbf{default game}, and
\[
  (G^p,S^{G^p})
\]
a \textbf{default mechanism}.

Thus, a default mechanism is obtained in three steps: we begin with a
candidate strategy profile, prune the game relative to that profile, and
then verify that the restricted candidate strategy profile is a PBE of
the resulting pruned game.

For notational simplicity, once a default mechanism has been obtained,
we suppress the superscript \(p\) and write it as \((G,S^G)\).

\subsection{Leakage-Proofness}

Fix a default mechanism \((G,S^G)\). When evaluating \(S^G\) in a
leakage environment, at a private history \(h_i^t\) with public component
\(h^t\), we interpret the default strategy as prescribing
\(S_i^G(\theta_i)(h^t)\); that is, the player disregards the leaked
component of \(h_i^t\).

For every \(i\in N\), \(\theta_i\in\Theta_i\), and
\(\lambda_i\in\Lambda_i\), define
\[
  H_i^G(\theta_i,\lambda_i)
  :=
  \left\{
    h_i^t\in H_i(\lambda_i):
    \Pr(h_i^t\mid\theta_i,\lambda_i;S^G)>0
  \right\}.
\]
Thus, \(H_i^G(\theta_i,\lambda_i)\) is the set of private histories that
may arise conditional on player \(i\)'s own payoff and leakage types when
all players follow the default strategy profile.

\begin{definition}
  \label{def:leakage-proof}
  A default mechanism \((G,S^G)\) is \textbf{leakage-proof} if, for every
  finite leakage type space \(\mathcal L\), the leakage environment
  \(\Gamma=(G,\mathcal L)\) admits a PBE
  \((S^{\mathcal L},\mu^{\mathcal L})\) such that, for every
  \(i\in N\), \(\theta_i\in\Theta_i\), \(\lambda_i\in\Lambda_i\), and
  \(h_i^t\in H_i^G(\theta_i,\lambda_i)\),
  \[
    S_i^{\mathcal L}(\theta_i,\lambda_i)(h_i^t)
    =
    S_i^G(\theta_i)(h^t),
  \]
  where \(h^t\) is the public component of \(h_i^t\). We also say \(S^G\) is a leakage-proof equilibrium in \(G\).
\end{definition}

Leakage-proofness therefore requires that, for every specification of
uncertainty about leakage, there exists a PBE that agrees with the
default strategy at every private history that can arise under default
play conditional on the player's own payoff and leakage types. The
equilibrium may prescribe different behavior at other private histories.
Thus, a player who has followed the default strategy in the past
continues to do so after observing leaked concurrent actions.\footnote{By imposing leakage-proofness only at private histories reachable under the default strategy, we follow the literature on extensive-form games in imposing additional equilibrium restrictions only at relevant histories; see, for example, \citet{Pearce-RationalizableStrategicBehavior-1984}, \citet{Shimoji-Watson-ConditionalDominanceRationalizability-1998}, \citet{Li-ObviouslyStrategyProofMechanisms-2017}, and \citet{Pycia-Troyan-TheorySimplicityGames-2023}.}

\section{Interim Implementation}
\label{sec:implementation}

Following \citet{Manelli-Vincent-BayesianDominantStrategyImplementation-2010} and \citet*{Gershkov-Goeree-Kushnir-Moldovanu-Shi-EquivalenceBayesianDominant-2013}, we focus on interim implementation. In a leakage environment, player \(i\)'s interim payoff is her expected payoff conditional on her payoff type \(\theta_i\) and leakage type \(\lambda_i\). \textbf{Interim implementation under leakage} requires this payoff to depend on \(\theta_i\), but not on \(\lambda_i\). This is a natural robustness requirement because the designer cannot condition the rules of the game \(G\) on the realized leakage order; accordingly, the interim payoff of a given payoff type should be invariant to the information leakage the player experiences.

For a PBE \((S^{\mathcal L},\mu^{\mathcal L})\) of \((G,\mathcal L)\), define player \(i\)'s interim payoff conditional on \((\theta_i,\lambda_i)\) by
\begin{align*}
  U_i(\theta_i,\lambda_i;S^{\mathcal L}) := {} & \sum_{\theta_{-i}\in\Theta_{-i}}\sum_{\omega\in\mathcal O_N} f(\theta_{-i}\mid\theta_i)\nu_i(\omega\mid\lambda_i) \notag\\
  & \quad\times\sum_{z\in Z}u_i\bigl(g(z),(\theta_i,\theta_{-i})\bigr)\Pr\bigl(z\mid(\theta_i,\theta_{-i}),\omega;S^{\mathcal L}\bigr).
\end{align*}
Fix an interim-payoff schedule \(U^\star=\bigl(U_i^\star(\theta_i)\bigr)_{i\in N,\,\theta_i\in\Theta_i}\).

\begin{definition}
  \label{def:interim-implementation}
  A PBE \((S^{\mathcal L},\mu^{\mathcal L})\) of \((G,\mathcal L)\) \textbf{interim implements \(U^\star\)} if \(U_i(\theta_i,\lambda_i;S^{\mathcal L})=U_i^\star(\theta_i)\) for every \(i\in N\), \(\theta_i\in\Theta_i\), and \(\lambda_i\in\Lambda_i\). A game \(G\) \textbf{interim implements \(U^\star\) under leakage} if, for every finite leakage type space \(\mathcal L\), the leakage environment \((G,\mathcal L)\) admits a PBE that interim implements \(U^\star\).
\end{definition}

The equality in Definition~\ref{def:interim-implementation} is conditional on the player's leakage type rather than on the full leakage order. The leakage order \(\omega\) determines \(\lambda_i=\ell_i(\omega)\), but a leakage type may pool several orders that player \(i\) cannot distinguish. For a strategy profile \(S\) in the game without leakage, we suppress \(\lambda_i\) and \(\mathcal L\) and write \(U_i(\theta_i;S)\) for the analogous interim payoff.

We supplement interim implementation with an outcome requirement to connect our formulation to standard outcome-based notions of implementation \citep*{Dasgupta-Hammond-Maskin-ImplementationSocialChoice-1979}. The outcome requirement may be vacuous, in which case no restriction is imposed on the induced outcome rule. An \textbf{outcome requirement under leakage} specifies a common restriction on the outcome rule induced under every realized leakage order.

For any strategy profile \(S\) in a leakage environment, payoff-type profile \(\theta\), and leakage order \(\omega\), define the induced outcome lottery by
\[
  \varphi(x\mid\theta,\omega;S):=\sum_{z\in Z:,g(z)=x}\Pr(z\mid\theta,\omega;S),\qquad x\in X.
\]
For each \(\omega\), the induced outcome rule is \(\varphi(\cdot\mid\cdot,\omega;S)\in\Delta(X)^\Theta\). For a strategy profile \(S\) in the game without leakage, we write \(\varphi(\cdot\mid\cdot;S)\) for the analogous outcome rule. We equip \(\Delta(X)^\Theta\) with its usual Euclidean topology.

\begin{definition}
  \label{def:outcome-requirement}
  An \textbf{outcome requirement} is a nonempty closed set \(\mathcal C\subseteq\Delta(X)^\Theta\). A PBE \((S^{\mathcal L},\mu^{\mathcal L})\) of \((G,\mathcal L)\) \textbf{satisfies \(\mathcal C\)} if \(\varphi(\cdot\mid\cdot,\omega;S^{\mathcal L})\in\mathcal C\) for every \(\omega\in\operatorname{supp}\nu\). A game \(G\) \textbf{satisfies \(\mathcal C\) under leakage} if, for every finite leakage type space \(\mathcal L\), the leakage environment \((G,\mathcal L)\) admits a PBE that satisfies \(\mathcal C\).
\end{definition}

The vacuous outcome requirement is \(\mathcal C=\Delta(X)^\Theta\), which imposes no restriction on the induced outcome rule. The same outcome requirement \(\mathcal C\) is imposed at every realized leakage order, but the induced outcome rule need not itself be invariant across orders: different leakage orders may induce different elements of \(\mathcal C\). Throughout, when we say that \(G\) interim implements \(U^\star\) and satisfies \(\mathcal C\) under leakage, we mean that, for every finite leakage type space \(\mathcal L\), there exists a single PBE of \((G,\mathcal L)\) that both interim implements \(U^\star\) and satisfies \(\mathcal C\).

\begin{theorem}[Interim implementation under leakage]
  \label{thm:implementation}
  Fix an interim-payoff schedule \(U^\star\) and an outcome requirement \(\mathcal C\).
  \begin{enumerate}
    \item If \((G,S^G)\) is a leakage-proof default mechanism that interim implements \(U^\star\) and satisfies \(\mathcal C\), then \(G\) interim implements \(U^\star\) and satisfies \(\mathcal C\) under leakage.
    \item If \(G\) interim implements \(U^\star\) and satisfies \(\mathcal C\) under leakage, then there exists a candidate strategy profile \(S^G\) such that, letting \(G^p\) denote the game pruned relative to \(S^G\), \((G^p,S^{G^p})\) is a leakage-proof default mechanism that interim implements \(U^\star\) and satisfies \(\mathcal C\).
  \end{enumerate}
\end{theorem}

\paragraph{Proof Sketch.}
Part~\textup{(1)} follows directly from leakage-proofness (Definition~\ref{def:leakage-proof}). For part~\textup{(2)}, the proof proceeds in four steps. First, we consider a sequence of ``rare-leakage'' environments in which leakage becomes arbitrarily unlikely, but whenever it occurs, one player is uniquely fully informed and observes all concurrent actions. From the equilibrium strategies of players who observe no concurrent actions, we obtain a limiting no-leakage strategy profile \(\overline S\). Second, interim implementation pins down the interim payoffs induced by \(\overline S\) at \(U^\star\), while closedness of \(\mathcal C\) ensures that the outcome requirement is preserved in the limit. Third, we show that no player can profitably condition her continuation behavior on leaked information against \(\overline S\): otherwise, a uniquely fully informed player in a sufficiently rare-leakage environment could obtain strictly more than \(U_i^\star(\theta_i)\), contradicting sequential rationality and interim implementation. Finally, we prune the original game relative to \(\overline S\) and complete behavior at the remaining off-path histories to construct a leakage-proof default mechanism. The full argument is given in Appendix~\ref{app:proof-implementation}.

Theorem~\ref{thm:implementation} separates the roles of interim implementation and the outcome requirement. Interim implementation drives the incentive argument: a uniquely fully informed player's equilibrium payoff remains fixed at \(U_i^\star(\theta_i)\), leaving no room for a strictly profitable use of leaked information. The outcome requirement plays no role in this comparison; its closedness is used only to ensure that the limiting outcome rule remains in \(\mathcal C\).

For example, let \(\chi:\Theta\to X\) be a social choice function. It induces the outcome requirement

\[
  \mathcal C_\chi:=\left\{\varphi\in\Delta(X)^\Theta:\varphi\bigl(\chi(\theta)\mid\theta\bigr)=1\text{ for every }\theta\in\Theta\right\}.
\]

Because \(\chi\) specifies a complete outcome at every payoff-type profile, it also determines the interim-payoff schedule

\[
  U_i^\chi(\theta_i):=\sum_{\theta_{-i}\in\Theta_{-i}}f(\theta_{-i}\mid\theta_i)u_i\bigl(\chi(\theta_i,\theta_{-i}),(\theta_i,\theta_{-i})\bigr).
\]

Theorem~\ref{thm:implementation} therefore applies with \(U^\star=U^\chi\) and \(\mathcal C=\mathcal C_\chi\).

More generally, let \(\mathcal X:\Theta\rightrightarrows X\) be a social choice correspondence. It induces the outcome requirement

\[
  \mathcal C_{\mathcal X}:=\left\{\varphi\in\Delta(X)^\Theta:\varphi\bigl(\mathcal X(\theta)\mid\theta\bigr)=1\text{ for every }\theta\in\Theta\right\}.
\]

Different leakage orders may induce different selections from \(\mathcal X(\theta)\), or different lotteries supported on \(\mathcal X(\theta)\), while satisfying the same outcome requirement. Unlike a social choice function, however, a social choice correspondence generally does not determine a unique interim-payoff schedule, because different admissible selections or lotteries may yield different utilities. Applying Theorem~\ref{thm:implementation} therefore requires an additional argument establishing a common interim-payoff schedule across leakage types.

This distinction leads naturally to the quasilinear environments studied next. There, fixing the allocation rule imposes a closed outcome requirement while leaving transfers, and hence interim utilities, unrestricted. Revenue maximization, together with incentive compatibility and individual rationality, can nevertheless identify an interim-rent schedule. Once this schedule is identified, Theorem~\ref{thm:implementation} converts interim implementation under leakage into a leakage-proof implementation of the fixed allocation rule.

\section{Quasilinear Environments}
\label{sec:quasilinear}

Let \(Y\) be a finite set of allocations. An outcome is \(x=(y,m)\in Y\times\mathbb R^n\), where \(m_i\) is player \(i\)'s payment to the designer. Player \(i\)'s utility is
\[
  u_i\bigl((y,m),\theta\bigr)
  =
  v_i(y,\theta)-m_i.
\]
The designer's revenue is normalized to
\[
  \sum_{i\in N}m_i,
\]
and every player's outside option is normalized to zero. Utilities and revenue are extended linearly to lotteries. Although \(Y\times\mathbb R^n\) need not be finite, every finite game \(G\) has a finite range \(g(Z)\); when applying Theorem~\ref{thm:implementation} to \(G\), we take the outcome set \(X\) to be this finite range.

Throughout this section, payoff types are independently distributed:
\[
  f(\theta)
  =
  \prod_{i\in N}f_i(\theta_i).
\]
Values may nevertheless be interdependent: \(v_i(y,\theta)\) may depend on the entire payoff-type profile. Write \(f_{-i}(\theta_{-i}):=\prod_{j\neq i}f_j(\theta_j)\).

For a strategy profile \(S\), define the induced allocation rule by
\[
  q^S(y\mid\theta,\omega)
  :=
  \sum_{z\in Z:\,y(g(z))=y}
  \Pr(z\mid\theta,\omega;S),
\]
where \(y(x)\) denotes the allocation component of outcome \(x\). For a PBE \((S^{\mathcal L},\mu^{\mathcal L})\), let
\begin{align*}
  M_i(\theta_i,\lambda_i;S^{\mathcal L})
  := {} &
  \sum_{\theta_{-i}\in\Theta_{-i}}
  \sum_{\omega\in\mathcal O_N}
  f_{-i}(\theta_{-i})
  \nu_i(\omega\mid\lambda_i)
  \notag\\
  &\quad\times
  \sum_{z\in Z}
  m_i(g(z))
  \Pr\bigl(z\mid(\theta_i,\theta_{-i}),\omega;S^{\mathcal L}\bigr)
\end{align*}
be player \(i\)'s interim expected payment. Let
\[
  \nu_i(\lambda_i)
  :=
  \sum_{\omega:\,\ell_i(\omega)=\lambda_i}\nu(\omega)
\]
denote the marginal probability of leakage type \(\lambda_i\).

The designer's expected revenue is
\begin{equation}
  \operatorname{Rev}(G,S^{\mathcal L})
  :=
  \sum_{i\in N}
  \sum_{\theta_i\in\Theta_i}
  \sum_{\lambda_i\in\Lambda_i}
  f_i(\theta_i)\nu_i(\lambda_i)
  M_i(\theta_i,\lambda_i;S^{\mathcal L}).
  \label{eq:designer-revenue}
\end{equation}
We impose \textbf{interim individual rationality} conditional on both of a player's types: \(U_i(\theta_i,\lambda_i;S^{\mathcal L})\geq0\) for every \(i\), \(\theta_i\), and \(\lambda_i\).

\subsection{Fixed Allocation}
\label{subsec:fixed-allocation}

Fix a stochastic allocation rule
\[
  q:\Theta\longrightarrow\Delta(Y).
\]
The rule does not depend on the leakage order. For \(\varphi\in\Delta(X)^\Theta\), let \(\varphi_Y\) denote its allocation marginal. For any fixed game whose range admits \(q\), define
\[
  \mathcal C_q
  :=
  \left\{
    \varphi\in\Delta(X)^\Theta:
    \varphi_Y(\,\cdot\mid\theta)=q(\,\cdot\mid\theta)
    \text{ for every }\theta\in\Theta
  \right\}.
\]
Thus, a PBE satisfies \(\mathcal C_q\) if it implements \(q\) conditional on every realized leakage order. Transfers remain unrestricted.

For \(\theta_i,\hat\theta_i\in\Theta_i\), define player \(i\)'s expected value when her true payoff type is \(\theta_i\) and she induces the allocation associated with \(\hat\theta_i\) by
\begin{align*}
  V_i(\theta_i,\hat\theta_i;q)
  :=
  \sum_{\theta_{-i}\in\Theta_{-i}}
  f_{-i}(\theta_{-i})
  \sum_{y\in Y}
  q(y\mid\hat\theta_i,\theta_{-i})
  v_i\bigl(y,(\theta_i,\theta_{-i})\bigr),
\end{align*}
and let
\[
  d_i(\theta_i,\hat\theta_i;q)
  :=
  V_i(\theta_i,\hat\theta_i;q)
  -
  V_i(\hat\theta_i,\hat\theta_i;q).
\]
Consider the complete directed graph on \(\Theta_i\), with the edge \(\hat\theta_i\to\theta_i\) assigned weight \(d_i(\theta_i,\hat\theta_i;q)\). The allocation rule \(q\) is \textbf{BIC implementable} if and only if this graph has no positive-weight directed cycle for any player \citep*{Rochet-NecessarySufficientCondition-1987,Malakhov-Vohra-SingleMultiDimensionalOptimal-2004,Muller-Perea-Wolf-WeakMonotonicityBayes-2007}.

Suppose \(q\) is BIC implementable. Add a source \(o_i\) and a zero-weight edge \(o_i\to\theta_i\) for every \(\theta_i\in\Theta_i\). Define player \(i\)'s \textbf{canonical rent} by
\begin{equation}
  R_i^q(\theta_i)
  :=
  \max_{\gamma:o_i\leadsto\theta_i}
  \operatorname{weight}(\gamma).
  \label{eq:canonical-fixed-allocation-rent}
\end{equation}
The absence of positive cycles makes this maximum finite. Define the corresponding revenue bound by
\begin{equation}
  \overline{\operatorname{Rev}}(q)
  :=
  \sum_{i\in N}
  \sum_{\theta_i\in\Theta_i}
  f_i(\theta_i)
  \left[
    V_i(\theta_i,\theta_i;q)
    -
    R_i^q(\theta_i)
  \right].
  \label{eq:canonical-fixed-allocation-revenue}
\end{equation}

\begin{lemma}[Fixed-allocation revenue bound]
  \label{lem:fixed-allocation-frontier}
  Let \(q\) be BIC implementable. Fix a finite leakage type space \(\mathcal L\). Every interim-IR PBE \((S^{\mathcal L},\mu^{\mathcal L})\) that satisfies \(\mathcal C_q\) obeys
  \begin{equation}
    U_i(\theta_i,\lambda_i;S^{\mathcal L})
    \geq
    R_i^q(\theta_i)
    \qquad
    \text{for every }i,\theta_i,\lambda_i,
    \label{eq:conditional-canonical-rent-bound}
  \end{equation}
  and
  \begin{align}
    \overline{\operatorname{Rev}}(q)
    -
    \operatorname{Rev}(G,S^{\mathcal L})
    = {} &
    \sum_{i\in N}
    \sum_{\lambda_i\in\Lambda_i}
    \nu_i(\lambda_i)
    \sum_{\theta_i\in\Theta_i}
    f_i(\theta_i)
    \notag\\
    &\quad\times
    \left[
      U_i(\theta_i,\lambda_i;S^{\mathcal L})
      -
      R_i^q(\theta_i)
    \right]
    \geq0.
    \label{eq:fixed-allocation-revenue-gap}
  \end{align}
  Equality holds if and only if \(U_i(\theta_i,\lambda_i;S^{\mathcal L})=R_i^q(\theta_i)\) for every \(i\), \(\theta_i\), and \(\lambda_i\).
\end{lemma}

Lemma~\ref{lem:fixed-allocation-frontier} shows that revenue maximization conditional on \(q\) identifies an interim-payoff schedule even though the allocation rule \(q\) alone does not. In particular, the revenue bound is attained if and only if every payoff type receives its canonical rent conditional on every leakage type. Moreover, no interim-IR equilibrium implementing \(q\) can yield revenue exceeding \(\overline{\operatorname{Rev}}(q)\). The proof is relegated to Appendix~\ref{app:fixed-allocation-frontier}.

Having identified the interim-payoff schedule, we can apply Theorem~\ref{thm:implementation} to obtain the following result.

\begin{definition}[Revenue maximization conditional on an allocation rule]
  \label{def:fixed-allocation-revenue-maximization}
  A finite game \(G\) \textbf{maximizes revenue conditional on \(q\) under leakage} if, for every finite leakage type space \(\mathcal L\), the leakage environment \((G,\mathcal L)\) admits an interim-IR PBE \((S^{\mathcal L},\mu^{\mathcal L})\) that satisfies \(\mathcal C_q\) and yields
  \[
    \operatorname{Rev}(G,S^{\mathcal L})
    =
    \overline{\operatorname{Rev}}(q).
  \]
\end{definition}

\begin{theorem}[Fixed-allocation revenue under leakage]
  \label{thm:fixed-allocation-leakage-proofness}
  Let \(q\) be BIC implementable. If \(G\) maximizes revenue conditional on \(q\) under leakage, then there exists a candidate strategy profile \(S^G\) such that, letting \(G^p\) be the game pruned relative to \(S^G\), \((G^p,S^{G^p})\) is a leakage-proof default mechanism that satisfies \(\mathcal C_q\), interim implements \(R^q\), and yields revenue \(\overline{\operatorname{Rev}}(q)\).

  Conversely, every leakage-proof default mechanism that satisfies \(\mathcal C_q\) and interim implements \(R^q\) maximizes revenue conditional on \(q\) under leakage.
\end{theorem}

\subsection{Single-Unit Auctions: Efficiency and Revenue}
\label{subsec:single-unit-auctions}

We now specialize to a single indivisible object. The allocation set is \(Y=N\cup\{0\}\), where allocation \(0\) means that the seller keeps the object and allocation \(i\in N\) means that bidder \(i\) receives it. Bidder \(i\)'s payoff type is her private value, and \(v_i(y,\theta)=\theta_i\mathbf 1_{\{y=i\}}\). The seller's valuation is normalized to zero, and her payoff is revenue. For each bidder \(i\), write \(\Theta_i=\{\theta_{i1}<\theta_{i2}<\cdots<\theta_{iK_i}\}\subseteq\mathbb R_+\).

For an allocation rule \(q\), write \(q_i(\theta):=q(i\mid\theta)\) and define bidder \(i\)'s interim allocation probability by
\[
  Q_i^q(\theta_{is}):=\sum_{\theta_{-i}\in\Theta_{-i}}f_{-i}(\theta_{-i})q_i(\theta_{is},\theta_{-i}).
\]
In this single-unit auction, \(q\) is BIC implementable if and only if \(Q_i^q\) is nondecreasing for every bidder \citep{Myerson-OptimalAuctionDesign-1981}. The canonical rents reduce to
\begin{equation}
  R_i^q(\theta_{i1})=0,\qquad R_i^q(\theta_{is})=\sum_{r=1}^{s-1}(\theta_{i,r+1}-\theta_{ir})Q_i^q(\theta_{ir})\quad\text{for }s\geq2.
  \label{eq:single-parameter-canonical-rents}
\end{equation}

\paragraph{Efficiency.}

An allocation rule \(q^E:\Theta\to\Delta(N\cup\{0\})\) is efficient if, for every \(\theta\in\Theta\), \(q_i^E(\theta)>0\) implies \(\theta_i=\max_{j\in N}\theta_j\) for every \(i\in N\), while \(q_0^E(\theta)>0\) implies \(\max_{j\in N}\theta_j=0\). For every fixed \(\theta_{-i}\), \(q_i^E(\theta_i,\theta_{-i})\) is nondecreasing in \(\theta_i\), so every efficient allocation rule is BIC implementable \citep{Myerson-OptimalAuctionDesign-1981}. The following result is a direct application of Theorem~\ref{thm:fixed-allocation-leakage-proofness} to efficient allocation rules.

\begin{corollary}[Efficient auctions under leakage]
  \label{cor:efficient-auctions-leakage}
  Fix an efficient allocation rule \(q^E\). If a finite auction \(G\) maximizes revenue conditional on \(q^E\) under leakage, then there exists a candidate strategy profile \(S^G\) such that, letting \(G^p\) be the game pruned relative to \(S^G\), \((G^p,S^{G^p})\) is a leakage-proof default mechanism that satisfies \(\mathcal C_{q^E}\), interim implements \(R^{q^E}\), and yields revenue \(\overline{\operatorname{Rev}}(q^E)\).

  Conversely, every leakage-proof default mechanism that satisfies \(\mathcal C_{q^E}\) and interim implements \(R^{q^E}\) maximizes revenue conditional on \(q^E\) under leakage.
\end{corollary}

Choosing \(q^E\) according to a fixed tie-breaking rule, the second-price auction implements \(q^E\) in weakly dominant strategies. Truthful bidding remains weakly dominant after any opponents' concurrent bids are observed. It is therefore leakage-proof and efficient, and it maximizes revenue conditional on \(q^E\) under leakage.

However, an analogous result does not follow directly for the revenue-maximization problem, because we do not yet have a characterization of the revenue-maximizing allocation rule under leakage.

\paragraph{Revenue Maximization.}
For a finite leakage type space \(\mathcal L\), define the optimal revenue benchmark by
\begin{equation}
  \operatorname{Rev}^\star(\mathcal L):=\sup_{\substack{\widehat G\text{ a finite single-unit auction},\\(\widehat S,\widehat\mu)\text{ an interim-IR PBE of }(\widehat G,\mathcal L)}}\operatorname{Rev}(\widehat G,\widehat S).
  \label{eq:revenue-benchmark-under-leakage}
\end{equation}
A finite auction \(G\) \textbf{maximizes revenue under leakage} if, for every finite leakage type space \(\mathcal L\), \((G,\mathcal L)\) has an interim-IR PBE attaining \(\operatorname{Rev}^\star(\mathcal L)\). Revenue is evaluated ex ante under the common prior \(f\otimes\nu\), rather than separately at each realized leakage order.

For a PBE \((S^{\mathcal L},\mu^{\mathcal L})\), write \(q_i^{S^{\mathcal L}}(\theta,\omega):=q^{S^{\mathcal L}}(i\mid\theta,\omega)\) and define
\begin{equation}
  Q_i(\theta_{is},\lambda_i;S^{\mathcal L}):=\sum_{\theta_{-i}\in\Theta_{-i}}\sum_{\omega\in\mathcal O_N}f_{-i}(\theta_{-i})\nu_i(\omega\mid\lambda_i)q_i^{S^{\mathcal L}}\bigl((\theta_{is},\theta_{-i}),\omega\bigr).
  \label{eq:auction-interim-allocation}
\end{equation}
Quasilinearity gives \(U_i(\theta_{is},\lambda_i;S^{\mathcal L})=\theta_{is}Q_i(\theta_{is},\lambda_i;S^{\mathcal L})-M_i(\theta_{is},\lambda_i;S^{\mathcal L})\).

Let \(F_i(\theta_{is}):=\sum_{r=1}^s f_i(\theta_{ir})\). The finite-type virtual value of \(\theta_{is}\) is
\begin{equation}
  \phi_i(\theta_{is}):=
  \begin{cases}\displaystyle \theta_{is}-(\theta_{i,s+1}-\theta_{is})\frac{1-F_i(\theta_{is})}{f_i(\theta_{is})},&s<K_i,\\[8pt]\theta_{iK_i},&s=K_i.
  \end{cases}
  \label{eq:finite-virtual-value}
\end{equation}

For an arbitrary value distribution, let \(\overline\phi_i\) denote bidder \(i\)'s ironed virtual value; see Appendix~\ref{app:auction-revenue-bound} for details. If \(\phi_i(\theta_{is})\) is strictly increasing in \(s\), so that the distribution is regular, then \(\overline\phi_i=\phi_i\).

\begin{lemma}[Revenue bound]
  \label{lem:auction-revenue-bound}
  Every interim-IR PBE \((S^{\mathcal L},\mu^{\mathcal L})\) of a finite single-unit auction satisfies
  \begin{align}
    \operatorname{Rev}(G,S^{\mathcal L})&\leq\mathbb E_{f\otimes\nu}\left[\sum_{i\in N}\overline \phi_i(\theta_i)q_i^{S^{\mathcal L}}(\theta,\omega)\right]\notag\\
    &\leq\mathbb E_f\left[\max\left\{0,\max_{i\in N}\overline \phi_i(\theta_i)\right\}\right].
    \label{eq:auction-revenue-bound}
  \end{align}
  Both inequalities are equalities if and only if the following conditions hold:
  \begin{enumerate}
    \item For every \(i\in N\) and \(\lambda_i\in\Lambda_i\), \(U_i(\theta_{i1},\lambda_i;S^{\mathcal L})=0\), and, for every \(s=2,\dots,K_i\),
      \begin{equation}
        U_i(\theta_{is},\lambda_i;S^{\mathcal L})-U_i(\theta_{i,s-1},\lambda_i;S^{\mathcal L})=(\theta_{is}-\theta_{i,s-1})Q_i(\theta_{i,s-1},\lambda_i;S^{\mathcal L}).
        \label{eq:downward-ic-binding}
      \end{equation}
    \item The interim allocations satisfy the ironing conditions; see Appendix~\ref{app:auction-revenue-bound} for details.
    \item For every \(\theta\in\Theta\) and \(\omega\in\operatorname{supp}\nu\), \(q_i^{S^{\mathcal L}}(\theta,\omega)>0\) only if \(\overline \phi_i(\theta_i)\geq0\) and \(\overline \phi_i(\theta_i)=\max_{j\in N}\overline \phi_j(\theta_j)\), while \(q_0^{S^{\mathcal L}}(\theta,\omega)>0\) only if \(\max_{j\in N}\overline \phi_j(\theta_j)\leq0\).
  \end{enumerate}
\end{lemma}

Lemma~\ref{lem:auction-revenue-bound} characterizes revenue maximization under leakage. The allocation rule must satisfy the ironing conditions and maximize ironed virtual surplus conditional on every leakage order, while the lowest payoff type's interim rent must be zero and every adjacent downward incentive constraint must bind conditional on every leakage type.

Let \(q^{\mathrm R}:\Theta\to\Delta(N\cup\{0\})\) denote a revenue-maximizing allocation rule whose interim allocation probabilities satisfy the ironing conditions and that allocates the object only among bidders whose ironed virtual values are nonnegative and maximal. Importantly, \(q^{\mathrm R}\) depends only on the payoff-type profile \(\theta\) and is therefore independent of the realized leakage order. Define
\begin{equation}
  \operatorname{Rev}^{\mathrm M}:=
  \mathbb E_f\left[\max\left\{0,\max_{i\in N}\overline\phi_i(\theta_i)\right\}\right].
\end{equation}

Having identified a revenue-maximizing allocation rule, we can characterize optimal auctions under leakage.

\begin{proposition}[Optimal auctions under leakage]
  \label{prop:optimal-revenue-arbitrary-distributions}
  Fix any revenue-maximizing allocation rule \(q^{\mathrm R}\). For every finite leakage type space \(\mathcal L\), \(\operatorname{Rev}^\star(\mathcal L)=\operatorname{Rev}^{\mathrm M}\).

  If a finite auction \(G\) maximizes revenue conditional on \(q^{\mathrm R}\) under leakage, then there exists a candidate strategy profile \(S^G\) such that, letting \(G^p\) be the game pruned relative to \(S^G\), \((G^p,S^{G^p})\) is a leakage-proof default mechanism that satisfies \(\mathcal C_{q^{\mathrm R}}\), interim implements \(R^{q^{\mathrm R}}\), and yields revenue \(\operatorname{Rev}^{\mathrm M}\).

  Conversely, every leakage-proof default mechanism that satisfies \(\mathcal C_{q^{\mathrm R}}\) and interim implements \(R^{q^{\mathrm R}}\) maximizes revenue conditional on \(q^{\mathrm R}\) under leakage and therefore maximizes revenue under leakage.
\end{proposition}

When bidders' values are identically distributed and regular, the second-price auction with the optimal reserve and a fixed tie-breaking rule is leakage-proof and revenue maximizing under leakage.

When the value distribution is irregular, a second-price auction with a reserve generally does not implement the required ironing and need not maximize revenue. Appendix~\ref{app:dutch} selects a fixed-priority revenue-maximizing allocation rule and implements it through a Dutch auction. The resulting auction yields \(\operatorname{Rev}^{\mathrm M}\) and is leakage-proof because bidders never act concurrently.

The corresponding first-price auction does not have this robustness. Without leakage, the first-price and Dutch formulations are strategically equivalent. Under leakage, however, a bidder in the first-price auction may observe opponents' concurrent bids before choosing her own and condition her bid on those observations. The first-price auction therefore need not be leakage-proof, and the equivalence between the first-price and Dutch auctions breaks down under leakage.

\begin{remark}
  In both Corollary~\ref{cor:efficient-auctions-leakage} and Proposition~\ref{prop:optimal-revenue-arbitrary-distributions}, we fix an allocation rule in advance. This restriction is needed primarily to pin down tie-breaking. Without it, an auction could remain efficient or revenue-maximizing while implementing different tie-breaking rules under different leakage orders, causing the interim-payoff schedule to vary with the realized leakage order. Appendix~\ref{app:fixed-allocation-tie} provides an example.
\end{remark}

\subsection{Leakage-Proofness Beyond Dominant Strategies: Multi-Unit Auctions}
\label{sec:multiunit}

We now show that leakage-proof mechanisms can achieve strictly higher revenues than dominant-strategy mechanisms. We establish this insight for auctions with multiple identical units and multi-unit demand.

There are $L\geq 1$ identical units. We allow both free disposal, where some units may remain unallocated, and mandatory allocation, where all units must be allocated. The corresponding feasible allocation sets are

$$
Y^{FD}=\left\{a\in\{0,\ldots,L\}^n:
\sum_{i\in N}a_i\leq L\right\},
\quad
Y^{MA}=\left\{a\in\{0,\ldots,L\}^n:
\sum_{i\in N}a_i=L\right\}.
$$

Throughout, $Y\in\{Y^{FD},Y^{MA}\}$ denotes the feasible allocation set of the environment, where $a_i$ is the number of units allocated to bidder $i$. For $k\in\{0,\ldots,L\}$, let

$$
y(k)=\bigl(\mathbf 1\{k\geq 1\},\ldots,\mathbf 1\{k\geq L\}\bigr)
$$
denote the associated tail-allocation vector. Bidder $i$'s type is a vector of marginal values $\theta_i=(\theta_{i1},\ldots,\theta_{iL})$, and her gross value from allocation $a$ is $v_i(a,\theta)=\theta_i\cdot y(a_i)$.

\begin{proposition}[Leakage-proofness beyond dominant strategies]
  \label{prop:multiunit-beyond-dic}
  There are multi-unit auction environments in which the optimal BIC revenue is strictly greater than the optimal DIC revenue, but the BIC-optimal revenue is attainable by a leakage-proof mechanism.
\end{proposition}

The proof considers environments in which the BIC-optimal revenue is attainable by a public sequential mechanism. By a public sequential mechanism, we mean a finite public perfect-information mechanism in which only one bidder acts at each stage. All previous actions are publicly observed, but no actions are concurrent. Such a mechanism is automatically robust to information leakage. Hence, once a public sequential mechanism has a PBE delivering a given allocation and payment rule in the no-leakage environment, an analogous equilibrium exists under every leakage type space.

More specifically, to establish Proposition~\ref{prop:multiunit-beyond-dic}, consider two bidders and two identical units with mandatory allocation, $Y=Y^{MA}$, as in \citet*{Gershkov-Goeree-Kushnir-Moldovanu-Shi-EquivalenceBayesianDominant-2013}. Each bidder's type is independently and uniformly distributed on
$$
\Theta_i=\left\{
  \theta^1=(1,1),\;
  \theta^2=(2,1),\;
  \theta^3=(5,3)
\right\}.
$$
Thus every type has decreasing marginal values and the types themselves are coordinatewise ordered. Nevertheless, BIC--DIC equivalence fails. As \citet*{Gershkov-Goeree-Kushnir-Moldovanu-Shi-EquivalenceBayesianDominant-2013} show, if both units must be allocated, the optimal revenues are $\mathrm{Rev}^{BIC}=191/45>38/9=\mathrm{Rev}^{DIC}$.

The public sequential mechanism elicits reports through four public binary tests. Bidder 1 is first asked whether she claims the highest type $\theta^3$ or passes. Bidder 2 is then asked the same question. Bidder 1, provided she has passed the first test, is subsequently asked whether she claims the middle type $\theta^2$ or passes; finally, bidder 2 is asked the analogous question. The allocation and payment rule at the terminal histories is chosen so that these choices are sequentially optimal. The resulting expected revenue is $191/45$, exactly the optimal BIC revenue. The allocation and payment rule and the full set of sequential incentive comparisons are given in the Appendix \ref{app:two-unit-leakage-proof}.\footnote{With two bidders, two identical units, and three iid coordinatewise ordered types with decreasing marginal values, the maximal BIC revenue can always be attained by a pure-strategy public sequential mechanism, whether or not free disposal is allowed. Thus, throughout this three-type class, BIC and leakage-proof implementation yield the same maximal revenue. This equivalence is fragile: it can fail with four types, with three bidders and two units, or with two bidders and three units. With four types, mixed public strategies may sometimes recover the BIC optimum but need not.}

\paragraph{Efficiency.}

The preceding proposition concerns revenue maximization, for which dominant-strategy incentive compatibility can be restrictive. For efficiency, the situation is simpler. An allocation rule $q^E$ is efficient if, at every payoff-type profile $\theta$, it assigns positive probability only to allocations $a$ that maximize
$$
\sum_{i\in N}\theta_i\cdot y(a_i).
$$
Every efficient allocation rule is dominant-strategy implementable by a VCG mechanism and hence BIC-implementable \citep[e.g.,][]{Nisan-IntroductionMechanismDesign-2007}. Theorem~\ref{thm:fixed-allocation-leakage-proofness} therefore yields the following analogue of Corollary~\ref{cor:efficient-auctions-leakage}.

\begin{corollary}[Efficient multi-unit auctions under leakage]
  \label{cor:efficient-multiunit-leakage}
  Fix an efficient allocation rule $q^E$. If a finite multi-unit auction $G$ maximizes revenue conditional on $q^E$ under leakage, then there exists a candidate strategy profile $S^G$ such that, letting $G^p$ be the game pruned relative to $S^G$, $(G^p,S^{G^p})$ is a leakage-proof default mechanism that satisfies $\mathcal C_{q^E}$, interim implements $R^{q^E}$, and yields revenue $\overline{\operatorname{Rev}}(q^E)$.

  Conversely, every leakage-proof default mechanism that satisfies $\mathcal C_{q^E}$ and interim implements $R^{q^E}$ maximizes revenue conditional on $q^E$ under leakage.
\end{corollary}

VCG provides a leakage-proof implementation of any efficient allocation rule, but need not maximize revenue conditional on that rule. The corollary shows that if an efficient mechanism attains the fixed-allocation revenue bound under leakage, it admits a leakage-proof implementation that attains the same bound.

\paragraph{Revenue Maximization.}
Finally, we consider conditions under which the BIC-optimal revenue can be attained by a dominant-strategy mechanism and, hence, by a leakage-proof multi-unit auction. For every bidder $i$, write
$\Theta_i=\{\theta_i^1,\ldots,\theta_i^{K_i}\}$,
$f_i^s=f_i(\theta_i^s)$, and
$F_i^s=\sum_{r=1}^s f_i^r$.
For $s<K_i$, define the virtual marginal-value vector by
\begin{equation}\label{eq:multiunit-virtual}
  \phi_i^s
  =
  \theta_i^s
  -
  \frac{1-F_i^s}{f_i^s}
  \bigl(\theta_i^{s+1}-\theta_i^s\bigr),
\end{equation}
and let $\phi_i^{K_i}=\theta_i^{K_i}$.
Write $\phi_i(\theta_i^s)=\phi_i^s$.

\begin{proposition}[Optimal multi-unit auctions under leakage]
  \label{prop:multiunit-optimal}
  Suppose that, for every bidder $i$, both the marginal-value vectors $\theta_i^1,\ldots,\theta_i^{K_i}$ and the virtual marginal-value vectors $\phi_i^1,\ldots,\phi_i^{K_i}$ are coordinatewise non-decreasing. Then the optimal BIC revenue is attainable by a leakage-proof multi-unit auction.
\end{proposition}

Proposition~\ref{prop:multiunit-optimal} extends Proposition~\ref{prop:optimal-revenue-arbitrary-distributions} to multi-unit environments with ordered marginal values. The ordering assumptions allow us to show directly that a revenue-maximizing allocation rule can be implemented in dominant strategies while attaining the optimal BIC revenue. This establishes BIC--DIC equivalence in terms of optimal revenue; leakage-proofness then follows immediately. In particular, the optimal revenue in the multi-unit environment is
\begin{equation}
  \label{eq:multiunit-optimal-revenue}
  \mathrm{Rev}^{MU}
  =
  \mathbf E_f\left[
    \max_{a\in Y}
    \sum_{i\in N}\phi_i(\theta_i)\cdot y(a_i)
  \right].
\end{equation}

In contrast, Proposition~\ref{prop:optimal-revenue-arbitrary-distributions} also characterizes revenue-maximizing implementations conditional on arbitrary BIC-optimal allocation rules, for which we rely on Theorem~\ref{thm:fixed-allocation-leakage-proofness}.

\section{Discussion}
\label{sec:discussion}

\subsection{Sequential Posterior Implementation}
Leakage-proofness is closely related to posterior implementation \citep{Green-Laffont-PosteriorImplementabilityTwoPerson-1987},  where they require equilibrium messages to remain mutually optimal after the information conveyed by the realized equilibrium message profile has been incorporated into players' posteriors. In our dynamic environment, the analogous requirement should be imposed after every public history and every realization of concurrent actions that a player might observe. We call this requirement \textbf{sequential posterior implementation}. We formalize the notion below and show that it is equivalent to leakage-proofness.

Fix a default mechanism \((G,S^G)\) and a player \(i\). A full-observation history for player \(i\) at stage \(t\) is
\[
  \widehat h_i^t
  =
  \left(h^t,a^t_{P(h^t)\setminus\{i\}}\right).
\]
Let \(\widehat H_i^G(\theta_i)\) be the set of such histories reached with positive probability under \(S^G\), conditional on \(\theta_i\). At \(\widehat h_i^t\in\widehat H_i^G(\theta_i)\), let
\[
  \widehat\mu_i^G\left(\theta_{-i}\mid\theta_i,\widehat h_i^t\right)
  =
  \frac{
    f(\theta_{-i}\mid\theta_i)
    \Pr\left(\widehat h_i^t\mid(\theta_i,\theta_{-i});S^G\right)
  }{
    \displaystyle
    \sum_{\hat\theta_{-i}\in\Theta_{-i}}
    f(\hat\theta_{-i}\mid\theta_i)
    \Pr\left(\widehat h_i^t\mid(\theta_i,\hat\theta_{-i});S^G\right)
  }
\]
be the posterior induced by default play. Let \(\widehat\Sigma_i(\widehat h_i^t)\) denote the continuation strategies that may condition, at every future stage, on the public history and all concurrent actions of the other active players. Extend the default strategy to full-observation histories by ignoring the additional information:
\[
  S_i^{G,\circ}(\theta_i)(\widehat h_i^t)
  :=
  S_i^G(\theta_i)(h^t).
\]
For \(\widetilde S_i\in\widehat\Sigma_i(\widehat h_i^t)\), let \(\widehat{\Pr}_i(z\mid\theta,\widehat h_i^t;\widetilde S_i,S_{-i}^G)\) be the induced probability of terminal history \(z\) in player \(i\)'s full-observation deviation problem.

\begin{definition}[Sequential posterior equilibrium]
  \label{def:sequential-posterior-equilibrium}
  \(S^G\) is a \textbf{sequential posterior equilibrium} in \(G\) if, for every \(i\in N\), \(\theta_i\in\Theta_i\), and \(\widehat h_i^t\in\widehat H_i^G(\theta_i)\),
  \begin{align*}
    S_i^{G,\circ}(\theta_i)
    \in
    \arg\max_{\widetilde S_i\in\widehat\Sigma_i(\widehat h_i^t)}
    &\sum_{\theta_{-i}\in\Theta_{-i}}
    \widehat\mu_i^G\left(\theta_{-i}\mid\theta_i,\widehat h_i^t\right)\\
    &\quad\times
    \sum_{z\in Z}
    u_i\bigl(g(z),(\theta_i,\theta_{-i})\bigr)
    \widehat{\Pr}_i\left(
      z
      \mid
      (\theta_i,\theta_{-i}),
      \widehat h_i^t;
      \widetilde S_i,S_{-i}^G
    \right).
  \end{align*}
\end{definition}

Thus, after every default-reachable public history and every realization of all opponents' concurrent actions, the default continuation remains optimal under the posterior generated by those observations, even though the player may continue to observe all opponents' concurrent actions.

\begin{theorem}\label{thm:sequential-leakage}
  \(S^G\) is a leakage-proof equilibrium in \(G\) if and only if \(S^G\) is a sequential posterior equilibrium in \(G\).
\end{theorem}

Sequential posterior implementation is equivalent to leakage-proofness. Necessity follows by considering, one player at a time, a commonly known leakage order under which that player observes every opponent's concurrent action. Sufficiency follows because every private history generated by an arbitrary leakage structure is a coarsening of a full-observation history: the sequential optimality inequalities can be averaged over the unobserved concurrent actions. Appendix~\ref{app:proof-sequential-leakage} provides the full argument.

This equivalence gives a microfoundation for posterior implementation in dynamic environments.

\subsection{Ex-Post and Dominant Strategies}\label{sec:expost}

Leakage-proofness is related to ex-post incentive compatibility (EPIC), but the two notions are not ordered without further restrictions. Pure-strategy EPIC is sufficient for leakage-proofness. After a payoff-type profile is fixed, opponents' pure actions determine every leaked observation along the default path; any profitable leakage-contingent continuation could therefore be embedded in an ex-post profitable strategy deviation. Appendix~\ref{app:pure-epic-proof} gives the argument.

The implication fails for mixed-strategy EPIC mechanisms because leakage may reveal the realization of another player's randomization rather than her payoff type. Consider matching pennies with a degenerate payoff-type space:
\[
  \begin{array}{c|cc}
    & H & T\\
    \hline
    H & (1,-1) & (-1,1)\\
    T & (-1,1) & (1,-1).
  \end{array}
\]
The unique Nash equilibrium mixes uniformly and is EPIC because there is no payoff uncertainty. It is not leakage-proof: after observing player~2's concurrent action, player~1 matches it and obtains payoff \(1\), rather than the equilibrium payoff \(0\).

The converse implication also fails. A Dutch auction is leakage-proof because bidders never act concurrently, but is not EPIC.

Any dominant strategy incentive-compatible (DIC) mechanism is leakage-proof. In general, leakage-proofness is a much weaker condition than DIC. The distinction from DIC becomes economically consequential when BIC--DIC equivalence fails; see the discussion following Proposition \ref{prop:multiunit-beyond-dic}.

\subsection{Beliefs about Leakage Orders}\label{misspecification}

In this paper, we have assumed that all players share a common prior over leakage orders, and leakage-proofness is defined by requiring robustness across all such common priors. In this subsection, we consider two alternatives to this assumption.

\paragraph{Common knowledege of leakage orders.}
When the realized leakage order is common knowledge among players, they may coordinate on order-specific equilibria. Leakage-proofness instead fixes a single no-leakage default equilibrium and requires it to remain sequentially optimal under every leakage order. Appendix~\ref{app:known-order-counterexample} gives a two-player example. Each of the three deterministic orders---no leakage, player~2 more informed, and player~1 more informed---admits a PBE implementing the same social choice function. However, there is no equilibrium of the game that is robust to information leakage.

\paragraph{Heterogeneous priors.}

Theorem \ref{thm:implementation} straightforwardly extends to environments in which the players have heterogenous priors over the leakage orders.\footnote{Indeed, implementing an interim allocation in a way that is robust to all heterogenous priors requires that it is in particular robust to all common priors, and hence, by Theorem \ref{thm:implementation}, it is leakage-proof, which establishes a heterogenous-prior analogue of the first part of Theorem \ref{thm:implementation}. Furthermore, the leakage-proof equilibrium remains leakge-proof when the priors over leakage are heterogeneous, which establishes a heterogenous-prior analogue of the second part of Theorem \ref{thm:implementation}.}

The common-prior assumption over leakage orders is however essential for the revenue bound in auctions (Lemma~\ref{lem:auction-revenue-bound}). Without it, revenue may then exceed the Myersonian optimal revenue.

To isolate this issue transparently, consider a continuous-type example. The example lies outside the maintained finiteness assumption and is used only to illustrate the role of a common prior. There are two bidders in a first-price auction with reserve price \(1/2\). Values are independently distributed according to
\[
  F(\theta)=1-\left(\theta+\frac12\right)^{-2},
  \qquad
  \theta\in\left[\frac12,\infty\right).
\]
The density is \(f(\theta)=2(\theta+1/2)^{-3}\), so the virtual value is
\[
  \phi(\theta)
  =
  \theta-\frac{1-F(\theta)}{f(\theta)}
  =
  \frac{\theta-1/2}{2}.
\]
Hence the optimal reserve is the lower endpoint \(1/2\). Without leakage, the symmetric first-price equilibrium bid is
\[
  \beta(\theta)
  =
  \theta-
  \frac{\displaystyle\int_{1/2}^{\theta}F(x)\,dx}{F(\theta)}
  =
  \theta-
  \frac{(\theta+1/2)(\theta-1/2)}{\theta+3/2}
  =
  \frac{6\theta+1}{4\theta+6}.
\]

Now abandon the common-prior requirement and suppose that both bidders are \textbf{paranoid}: each dogmatically believes that she is less informed and that the opponent observes her bid before responding. These beliefs cannot be generated by a common prior over \(\mathcal O_N\). Under such a belief, a bidder who submits \(b\) assesses her winning probability as \(F(b)\) and solves
\[
  \max_{b\geq1/2}(\theta-b)F(b).
\]
Let \(b^*(\theta)\) denote the maximizer. A bid above \(\theta\) yields negative utility and is never optimal. On \([1/2,\theta)\), the derivative of the objective is
\[
  g_\theta(b)
  :=
  \frac{2(\theta-b)}{(b+1/2)^3}
  -
  \left(1-\frac{1}{(b+1/2)^2}\right),
\]
and the second derivative is strictly negative. The first-order condition has a unique solution: \(t(\theta):=b^*(\theta)+1/2\) is the unique solution to
\[
  t(\theta)^3+t(\theta)=2\theta+1,
\]
so \(b^*\) is strictly increasing. Evaluating the derivative at the no-leakage equilibrium bid gives
\[
  g_\theta\bigl(\beta(\theta)\bigr)
  =
  \frac{(2\theta-1)^2}{8(2\theta+1)}
  >0
  \qquad
  \text{for every \(\theta>1/2\)}.
\]
Strict concavity therefore implies \(b^*(\theta)>\beta(\theta)\) for every \(\theta>1/2\).

The Myersonian optimal revenue is
\[
  \operatorname{Rev}^{\mathrm M}
  =
  \mathbb E\left[\phi\left(\max\{\theta_1,\theta_2\}\right)\right]
  =
  \frac56.
\]
Because both \(\beta\) and \(b^*\) are increasing, the expected winning bid under paranoid behavior satisfies
\[
  \mathbb E\left[b^*\left(\max\{\theta_1,\theta_2\}\right)\right]
  >
  \mathbb E\left[\beta\left(\max\{\theta_1,\theta_2\}\right)\right]
  =
  \operatorname{Rev}^{\mathrm M}.
\]

Paranoid beliefs can arise under a common prior over leakage orders, but only when the prior is supported on orders in which one bidder is more informed than the other. Under such orders, the more informed bidder can condition her bid on the opponent's concurrent bid, reducing revenue below the Myersonian optimal revenue. Our revenue bound in Lemma~\ref{lem:auction-revenue-bound} shows that any revenue gain generated by paranoid beliefs cannot offset the revenue loss when one bidder is more informed.

\bibliographystyle{econ}
\bibliography{leakage}

\appendix

\section{Proof of Theorem~\ref{thm:implementation}}
\label{app:proof-implementation}

We first prove part~\textup{(1)}. Fix a finite leakage type space \(\mathcal L\), and let \((S^{\mathcal L},\mu^{\mathcal L})\) be the PBE supplied by leakage-proofness. At the initial history, every player follows the default strategy. If all players have followed the default strategy up to some stage, every private history reached at that stage is default-reachable, so leakage-proofness requires every player to continue following the default strategy. Induction over stages implies that, conditional on every \((\theta,\omega)\), equilibrium play induces the same terminal-history distribution as \(S^G\) in the game without leakage. Hence \(\varphi(\cdot\mid\cdot,\omega;S^{\mathcal L})=\varphi(\cdot\mid\cdot;S^G)\in\mathcal C\) for every \(\omega\in\operatorname{supp}\nu\). Because payoff types and leakage orders are independent, \(U_i(\theta_i,\lambda_i;S^{\mathcal L})=U_i(\theta_i;S^G)=U_i^\star(\theta_i)\) for every \(i\), \(\theta_i\), and \(\lambda_i\). This proves part~\textup{(1)}.

We now prove part~\textup{(2)}. The case \(n=1\) is immediate, so suppose \(n\geq2\).

\paragraph{Rare-leakage Environments.}
For every \(\eta\in(0,1)\), consider a leakage type space \(\mathcal L^\eta\) whose prior is supported on \(\{\omega^0,\omega^1,\dots,\omega^n\}\). At \(\omega^0\), all players are equally informed, so no player observes any concurrent action. At \(\omega^i\), player \(i\) is uniquely more informed than every other player and therefore observes every opponent's concurrent action, while every opponent observes no concurrent action. Let
\[
  \nu^\eta(\omega^0)=1-\eta,\qquad \nu^\eta(\omega^i)=\frac{\eta}{n}\quad\text{for every }i\in N.
\]
For each player \(i\), let \(\Lambda_i^\eta=\{\lambda_i^0,\lambda_i^1\}\), where
\[
  \ell_i^\eta(\omega^i)=\lambda_i^1,\qquad \ell_i^\eta(\omega^j)=\lambda_i^0\quad\text{for every }j\neq i,
\]
including \(j=0\). Thus, \(\lambda_i^1\) reveals that player \(i\) is uniquely fully informed, whereas \(\lambda_i^0\) never observes a concurrent action. Conditional on \(\lambda_i^0\),
\[
  \alpha^\eta:=\nu_i^\eta(\omega^0\mid\lambda_i^0)=\frac{1-\eta}{1-\eta/n}\longrightarrow1\qquad\text{as }\eta\downarrow0.
\]
Choose a sequence \(\eta_m\downarrow0\). By hypothesis, for every \(m\) there is a PBE \((S^m,\mu^m)\) of \((G,\mathcal L^{\eta_m})\) that interim implements \(U^\star\) and satisfies \(\mathcal C\). Define \(S_i^{m,0}(\theta_i):=S_i^m(\theta_i,\lambda_i^0)\). Because \(\lambda_i^0\) never observes a concurrent action, \(S_i^{m,0}(\theta_i)\) is naturally a behavioral strategy in the original game without leakage.

We equip the behavioral-strategy spaces in the original game with the topology induced by the sup norm. For player \(i\), define
\[
  \|S_i-\widetilde S_i\|_\infty:=\max_{\substack{\theta_i\in\Theta_i,\;h^t\in H\setminus Z,\;i\in P(h^t),\\a_i\in A_i(h^t)}}\left|S_i(\theta_i)(a_i\mid h^t)-\widetilde S_i(\theta_i)(a_i\mid h^t)\right|.
\]
For strategy profiles \(S=(S_i)_{i\in N}\) and \(\widetilde S=(\widetilde S_i)_{i\in N}\), define \(\|S-\widetilde S\|_\infty:=\max_{i\in N}\|S_i-\widetilde S_i\|_\infty\). Because the game and payoff-type spaces are finite, each player's behavioral-strategy space is a finite product of compact simplices. Hence the strategy-profile space is compact under \(\|\cdot\|_\infty\), and all convergence statements for strategies below refer to this norm.

\paragraph{The Limiting Strategy Profile.}
By compactness, after passing to a subsequence we may assume that \(S^{m,0}\to\overline S\). At \(\omega^0\), every player has leakage type \(\lambda_i^0\), so the conditional outcome rule induced by \(S^m\) is exactly the no-leakage outcome rule induced by \(S^{m,0}\). Hence \(\varphi(\cdot\mid\cdot;S^{m,0})=\varphi(\cdot\mid\cdot,\omega^0;S^m)\in\mathcal C\). Terminal-history probabilities are polynomial, and therefore continuous, functions of behavioral strategies in a finite extensive-form game. Thus \(\varphi(\cdot\mid\cdot;S^{m,0})\to\varphi(\cdot\mid\cdot;\overline S)\) in \(\Delta(X)^\Theta\), and closedness of \(\mathcal C\) gives
\begin{equation}
  \varphi(\cdot\mid\cdot;\overline S)\in\mathcal C.
  \label{eq:limit-outcome-requirement}
\end{equation}

We next show that \(\overline S\) induces the interim-payoff schedule \(U^\star\). Fix \(i\in N\) and \(\theta_i\in\Theta_i\). For \(j\in\{0,1,\dots,n\}\), let
\begin{align*}
  U_{i,m}^j(\theta_i):={}&\sum_{\theta_{-i}\in\Theta_{-i}}f(\theta_{-i}\mid\theta_i)\sum_{z\in Z}u_i\bigl(g(z),(\theta_i,\theta_{-i})\bigr)\\
  &\quad\times\Pr\bigl(z\mid(\theta_i,\theta_{-i}),\omega^j;S^m\bigr)
\end{align*}
be player \(i\)'s interim payoff conditional on \(\theta_i\) and \(\omega^j\), and define
\[
  \widetilde U_{i,m}(\theta_i):=\frac{1}{n-1}\sum_{j\in N\setminus\{i\}}U_{i,m}^j(\theta_i).
\]
Conditional on \(\lambda_i^0\), the order is \(\omega^0\) with probability \(\alpha^{\eta_m}\), and otherwise is uniformly distributed over \(\{\omega^j:j\in N\setminus\{i\}\}\). Interim implementation for \(\lambda_i^0\) therefore gives
\begin{equation}
  U_i^\star(\theta_i)=\alpha^{\eta_m}U_{i,m}^0(\theta_i)+\bigl(1-\alpha^{\eta_m}\bigr)\widetilde U_{i,m}(\theta_i).
  \label{eq:slow-payoff-mixture}
\end{equation}
Because utilities are bounded, choose \(B<\infty\) such that \(|u_i(x,\theta)|\leq B\) for every \(i\), \(x\), and \(\theta\). Equation~\eqref{eq:slow-payoff-mixture} implies
\begin{equation}
  \left|U_{i,m}^0(\theta_i)-U_i^\star(\theta_i)\right|\leq2B\bigl(1-\alpha^{\eta_m}\bigr)\longrightarrow0.
  \label{eq:no-leakage-payoff-convergence}
\end{equation}
At \(\omega^0\), play is generated by \(S^{m,0}\). Continuity of expected utility therefore yields
\begin{equation}
  U_i(\theta_i;\overline S)=U_i^\star(\theta_i)\qquad\text{for every }i\in N\text{ and }\theta_i\in\Theta_i.
  \label{eq:limit-interim-payoff}
\end{equation}

\paragraph{No Profitable Leakage-Contingent Continuation.}
Given a strategy profile \(S\) in the original game, define its \emph{leakage-blind extension} \(S^\circ\) by \(S_i^\circ(\theta_i,\lambda_i)(h_i^t):=S_i(\theta_i)(h^t)\), where \(h^t\) is the public component of \(h_i^t\). Thus, \(S_i^\circ\) ignores the leaked component of the private history. Until the default mechanism is constructed, call a private history \(h_i^t\) \emph{\(\overline S\)-reachable} for \((\theta_i,\lambda_i)\) if \(\Pr(h_i^t\mid\theta_i,\lambda_i;\overline S^\circ)>0\). We claim that, in every finite leakage type space, the continuation prescribed by \(\overline S_i^\circ\) is optimal at every \(\overline S\)-reachable private history against opponents following \(\overline S_{-i}^\circ\).

Suppose, toward a contradiction, that the claim fails. Then there exist a finite leakage type space \(\mathcal L\), a player \(i\), a payoff type \(\theta_i\), a leakage type \(\lambda_i\), an \(\overline S\)-reachable private history \(h_i^t\), and a continuation strategy \(\sigma_i\) that gives player \(i\) a strict continuation gain against \(\overline S_{-i}^\circ\). A \emph{full-observation refinement} of \(h_i^t\) supplements it with the concurrent actions of all other active players that are not observed under \(\lambda_i\). The gain at \(h_i^t\) is a finite weighted average of the gains conditional on these refinements, so some refinement \(\bar h_i^t\), reached with positive probability conditional on \(\theta_i\), preserves a strict gain. Because \(\overline S^\circ\) is leakage blind and payoff types are independent of the leakage order, the same refinement is reached with positive probability when player \(i\) fully observes concurrent actions and induces the same posterior over \(\theta_{-i}\).

Let \(\widehat\sigma_i\) be the full-observation extension of \(\sigma_i\): at every subsequent private history, player \(i\) discards the additional concurrent actions that would not be observed under \(\lambda_i\) and then follows \(\sigma_i\). A fully informed player can therefore implement \(\widehat\sigma_i\) by ignoring the additional information.

We show that the strict continuation gain at \(\bar h_i^t\) persists along \(S^{m,0}\to\overline S\). For each \(\theta_{-i}\in\Theta_{-i}\), let
\[
  r_m(\theta_{-i}):=\Pr\bigl(\bar h_i^t\mid(\theta_i,\theta_{-i}),\omega^i;(S^{m,0})^\circ\bigr),
\]
and define \(r(\theta_{-i})\) analogously using \(\overline S^\circ\). Reach probabilities are continuous in behavioral strategies, so \(r_m(\theta_{-i})\to r(\theta_{-i})\) for every \(\theta_{-i}\). Let
\[
  r_m:=\sum_{\theta_{-i}\in\Theta_{-i}}f(\theta_{-i}\mid\theta_i)r_m(\theta_{-i}),\qquad r:=\sum_{\theta_{-i}\in\Theta_{-i}}f(\theta_{-i}\mid\theta_i)r(\theta_{-i}).
\]
Because \(\bar h_i^t\) is reached with positive probability conditional on \(\theta_i\), \(r>0\). Hence, for all sufficiently large \(m\),
\begin{equation}
  r_m>\frac{r}{2}>0.
  \label{eq:reach-probability-bound}
\end{equation}

For any continuation strategy \(\kappa_i\) after \(\bar h_i^t\), let \(W_m(\kappa_i,\theta_{-i})\) denote player \(i\)'s expected continuation payoff conditional on \((\theta_i,\theta_{-i})\), \(\omega^i\), and \(\bar h_i^t\), when player \(i\) uses \(\kappa_i\) and the opponents follow \((S_{-i}^{m,0})^\circ\). Define \(W(\kappa_i,\theta_{-i})\) analogously using \(\overline S_{-i}^\circ\). Since the continuation game is finite, \(W_m(\widehat\sigma_i,\theta_{-i})\to W(\widehat\sigma_i,\theta_{-i})\) and, by joint continuity, \(W_m((S_i^{m,0})^\circ,\theta_{-i})\to W(\overline S_i^\circ,\theta_{-i})\).

For \(r_m>0\), define
\[
  V_i^m(\kappa_i\mid\bar h_i^t,\theta_i):=\frac{\sum_{\theta_{-i}\in\Theta_{-i}}f(\theta_{-i}\mid\theta_i)r_m(\theta_{-i})W_m(\kappa_i,\theta_{-i})}{r_m},
\]
and define \(V_i(\kappa_i\mid\bar h_i^t,\theta_i)\) analogously using \(r(\theta_{-i})\), \(r\), and \(W\). By the choice of \(\bar h_i^t\), \(\Delta:=V_i(\widehat\sigma_i\mid\bar h_i^t,\theta_i)-V_i(\overline S_i^\circ\mid\bar h_i^t,\theta_i)>0\). The numerators are finite sums of products of convergent terms and the denominators converge to \(r>0\). Therefore
\[
  V_i^m(\widehat\sigma_i\mid\bar h_i^t,\theta_i)-V_i^m((S_i^{m,0})^\circ\mid\bar h_i^t,\theta_i)\longrightarrow\Delta,
\]
so, for all sufficiently large \(m\),
\[
  V_i^m(\widehat\sigma_i\mid\bar h_i^t,\theta_i)-V_i^m((S_i^{m,0})^\circ\mid\bar h_i^t,\theta_i)>\frac{\Delta}{2}.
\]

Now consider player \(i\)'s uniquely informed leakage type \(\lambda_i^1\) in \(\mathcal L^{\eta_m}\). At \(\omega^i\), player \(i\) observes all concurrent actions, while every opponent \(j\neq i\) has leakage type \(\lambda_j^0\) and follows \(S_j^{m,0}\) in equilibrium. Consider the deviation by \((\theta_i,\lambda_i^1)\) that ignores leaked information and follows \(S_i^{m,0}\) until \(\bar h_i^t\), then switches to \(\widehat\sigma_i\). Relative to ignoring leakage and following \(S_i^{m,0}\) throughout, its gain conditional on \(\theta_i\) and \(\omega^i\) is at least \(r_m\Delta/2\), which by Equation~\eqref{eq:reach-probability-bound} is at least \(r\Delta/4>0\) for all sufficiently large \(m\).

If player \(i\) ignores her additional observations and follows \(S_i^{m,0}\) throughout at \(\omega^i\), every player uses \(S^{m,0}\), so the public distribution of play is the same as at \(\omega^0\). Her payoff from this leakage-blind strategy is therefore \(U_{i,m}^0(\theta_i)\), which converges to \(U_i^\star(\theta_i)\) by Equation~\eqref{eq:no-leakage-payoff-convergence}. The deviation constructed above consequently gives player \(i\) strictly more than \(U_i^\star(\theta_i)\) for all sufficiently large \(m\). On the other hand, \(\lambda_i^1\) identifies \(\omega^i\), and interim implementation requires \(U_i(\theta_i,\lambda_i^1;S^m)=U_i^\star(\theta_i)\). Thus the constructed deviation gives type \((\theta_i,\lambda_i^1)\) strictly more than her equilibrium payoff, contradicting sequential rationality of \((S^m,\mu^m)\).

Therefore, in every finite leakage type space, the \(\overline S_i^\circ\)-continuation is optimal at every \(\overline S\)-reachable private history against opponents following \(\overline S_{-i}^\circ\).

\paragraph{Pruning and Completion to a PBE.}
We now prune \(G\) relative to \(\overline S\) as in Definition~\ref{def:pruning}, and denote the resulting game by \(G^p\). Fix an arbitrary finite leakage type space \(\mathcal L\). For player \(i\) and payoff type \(\theta_i\), call a private history \(h_i^t\), with public component \(h^t\), \emph{own-compatible} if every past action of player \(i\) recorded in \(h^t\) is assigned positive probability by \(\overline S_i(\theta_i)\) at the corresponding predecessor public history.

By construction of \(G^p\), every feasible own-compatible private history is \(\overline S\)-reachable. Indeed, player \(i\)'s own recorded actions have positive probability by own-compatibility, while pruning ensures that the recorded actions of the other players are generated with positive probability by payoff types compatible with the preceding public histories. Full support of payoff types and leakage orders then gives the required positive reach probability. The preceding claim therefore implies that \(\overline S_i^\circ\) is optimal at every own-compatible private history against opponents following \(\overline S_{-i}^\circ\).

It remains to specify behavior at own-incompatible histories. Construct an auxiliary finite extensive-form game in which every move at an own-compatible private history is replaced by a chance move with distribution \(\overline S_i(\theta_i)(h^t)\), while decisions at own-incompatible histories remain unchanged. Let \((\widehat S^{\mathcal L},\widehat\mu^{\mathcal L})\) be a sequential equilibrium of this auxiliary game. In \((G^p,\mathcal L)\), prescribe \(\overline S_i^\circ\) at own-compatible histories, prescribe \(\widehat S_i^{\mathcal L}\) at own-incompatible histories, and use the belief system \(\widehat\mu^{\mathcal L}\).

This assessment is a PBE. At an own-compatible history, a unilateral deviation by player \(i\) does not affect the opponents' own-compatibility, so they continue to follow
\(\overline S_{-i}^\circ\); sequential rationality therefore follows from the no-profitable-continuation result above. At an own-incompatible history, own-incompatibility persists along all successor histories of player \(i\), so her continuation problem is exactly the one in the auxiliary game. Sequential rationality therefore follows from \((\widehat S^{\mathcal L},\widehat\mu^{\mathcal L})\).

Apply the same completion argument to the degenerate no-leakage environment. This yields a strategy profile on \(G^p\) that agrees with \(\overline S\) at every own-compatible history and is part of a PBE of \(G^p\). Extend this strategy profile arbitrarily to histories outside \(G^p\), and denote the resulting candidate strategy profile in \(G\) by \(S^G\). The game pruned relative to \(S^G\) is therefore precisely \(G^p\), and \(S^{G^p}=S^G|_{G^p}\). Hence \((G^p,S^{G^p})\) is a default mechanism.

For every finite leakage type space \(\mathcal L\), the PBE constructed above agrees with \(S^{G^p}\) at every default-reachable private history, so \((G^p,S^{G^p})\) is leakage-proof. Moreover, equilibrium play remains at own-compatible histories and therefore induces, conditional on every payoff-type profile and leakage order, the same terminal-history distribution as \(\overline S\). Equations~\eqref{eq:limit-interim-payoff} and \eqref{eq:limit-outcome-requirement} imply that the constructed PBE interim implements \(U^\star\) and satisfies \(\mathcal C\). This proves part~\textup{(2)}.

\section{Proofs for Section~\ref{sec:quasilinear}}
\label{app:quasilinear}

\subsection{Proof of Lemma~\ref{lem:fixed-allocation-frontier}}
\label{app:fixed-allocation-frontier}

Fix a finite leakage type space \(\mathcal L\) and an interim-IR PBE \((S^{\mathcal L},\mu^{\mathcal L})\) satisfying \(\mathcal C_q\). Fix a player \(i\) and leakage type \(\lambda_i\). For any two payoff types \(\theta_i,\hat\theta_i\in\Theta_i\), type \((\theta_i,\lambda_i)\) can imitate the complete behavioral strategy assigned to \((\hat\theta_i,\lambda_i)\). This deviation yields expected value \(V_i(\theta_i,\hat\theta_i;q)\) and the same expected payment as type \((\hat\theta_i,\lambda_i)\). Sequential rationality at the initial history therefore implies
\[
  U_i(\theta_i,\lambda_i;S^{\mathcal L})\geq U_i(\hat\theta_i,\lambda_i;S^{\mathcal L})+d_i(\theta_i,\hat\theta_i;q)
\]
for every \(\theta_i,\hat\theta_i\in\Theta_i\), while interim IR implies \(U_i(\theta_i,\lambda_i;S^{\mathcal L})\geq0\) for every \(\theta_i\in\Theta_i\).

Summing the above incentive inequalities around any directed cycle shows that every cycle has weakly nonpositive weight. Since the graph is finite, every walk can therefore be reduced to a simple path by deleting cycles without lowering its weight, so the maximum defining \(R_i^q(\theta_i)\) is finite and attained. The source edges imply \(R_i^q(\theta_i)\geq0\), while appending the edge \(\hat\theta_i\to\theta_i\) to a maximizing path to \(\hat\theta_i\), and deleting any resulting cycles, gives
\[
  R_i^q(\theta_i)\geq R_i^q(\hat\theta_i)+d_i(\theta_i,\hat\theta_i;q).
\]
Moreover, iterating the preceding incentive and individual-rationality inequalities along any path from the source to \(\theta_i\) implies that \(U_i(\theta_i,\lambda_i;S^{\mathcal L})\) is at least the weight of that path. Taking the maximum over paths yields
\[
  U_i(\theta_i,\lambda_i;S^{\mathcal L})\geq R_i^q(\theta_i),
\]
which proves Equation~\eqref{eq:conditional-canonical-rent-bound}.

Let \(M_i(\theta_i,\lambda_i;S^{\mathcal L})\) denote player \(i\)'s interim expected payment conditional on \((\theta_i,\lambda_i)\). Because the allocation rule is \(q\) conditional on every leakage order, quasilinearity gives
\[
  M_i(\theta_i,\lambda_i;S^{\mathcal L})=V_i(\theta_i,\theta_i;q)-U_i(\theta_i,\lambda_i;S^{\mathcal L}).
\]
Substituting this identity into Equation~\eqref{eq:designer-revenue} and subtracting from Equation~\eqref{eq:canonical-fixed-allocation-revenue} gives Equation~\eqref{eq:fixed-allocation-revenue-gap}. Every \(f_i(\theta_i)\) and \(\nu_i(\lambda_i)\) is strictly positive, and every bracketed term is nonnegative. Hence the revenue gap is zero if and only if every bracketed term is zero.

\subsection{Proof of Theorem~\ref{thm:fixed-allocation-leakage-proofness}}
\label{app:fixed-allocation-leakage-proofness}
Suppose first that \(G\) maximizes revenue conditional on \(q\) under leakage. For every finite leakage type space \(\mathcal L\), select the equilibrium supplied by Definition~\ref{def:fixed-allocation-revenue-maximization}. Its revenue is \(\overline{\operatorname{Rev}}(q)\), so Lemma~\ref{lem:fixed-allocation-frontier} gives
\[
  U_i(\theta_i,\lambda_i;S^{\mathcal L})
  =
  R_i^q(\theta_i)
\]
for every \(i\), \(\theta_i\), and \(\lambda_i\). The same equilibrium satisfies \(\mathcal C_q\). Hence \(G\) interim implements \(R^q\) and satisfies \(\mathcal C_q\) under leakage. Part~\textup{(2)} of Theorem~\ref{thm:implementation} yields the claimed leakage-proof pruned default mechanism. Its revenue is \(\overline{\operatorname{Rev}}(q)\) by Equation~\eqref{eq:fixed-allocation-revenue-gap}.

Conversely, suppose a leakage-proof default mechanism satisfies \(\mathcal C_q\) and interim implements \(R^q\). Part~\textup{(1)} of Theorem~\ref{thm:implementation} supplies, in every finite leakage type space, a PBE satisfying both properties. Equation~\eqref{eq:fixed-allocation-revenue-gap} gives revenue \(\overline{\operatorname{Rev}}(q)\), while Lemma~\ref{lem:fixed-allocation-frontier} shows that no interim-IR PBE implementing \(q\) conditional on every leakage order can yield more. The default mechanism therefore maximizes revenue conditional on \(q\) under leakage.

\subsection{Proof of Lemma~\ref{lem:auction-revenue-bound}}
\label{app:auction-revenue-bound}

Fix bidder \(i\), and define the upper-tail probabilities \(x_{is}:=\sum_{r=s}^{K_i}f_i(\theta_{ir})\) for \(s=1,\dots,K_i\), with \(x_{i,K_i+1}:=0\). Define the revenue curve at these knots by \(\mathcal R_i(x_{is}):=x_{is}\theta_{is}\) and \(\mathcal R_i(0):=0\), and extend it linearly between adjacent knots. Let \(\overline{\mathcal R}_i\) be the least concave majorant of \(\mathcal R_i\). The ironed virtual value of \(\theta_{is}\) is the slope
\begin{equation}
  \overline\phi_i(\theta_{is}):=\frac{\overline{\mathcal R}_i(x_{is})-\overline{\mathcal R}_i(x_{i,s+1})}{f_i(\theta_{is})}.
  \label{eq:ironed-virtual-value}
\end{equation}
The sequence \(\overline\phi_i(\theta_{i1}),\dots,\overline\phi_i(\theta_{iK_i})\) is weakly increasing. Under regularity, \(\mathcal R_i\) is concave and \(\overline\phi_i=\phi_i\).

Let \(D_{is}:=\overline{\mathcal R}_i(x_{is})-\mathcal R_i(x_{is})\), with \(D_{i1}=D_{i,K_i+1}=0\). For every nondecreasing sequence \((Q_{is})_{s=1}^{K_i}\), summation by parts gives
\begin{equation}
  \sum_{s=1}^{K_i}f_i(\theta_{is})\bigl(\overline\phi_i(\theta_{is})-\phi_i(\theta_{is})\bigr)Q_{is}=\sum_{s=2}^{K_i}D_{is}(Q_{is}-Q_{i,s-1})\geq0.
  \label{eq:ironing-identity}
\end{equation}
Equality holds if and only if \(Q_{is}=Q_{i,s-1}\) whenever \(D_{is}>0\). We refer to these equalities as the ironing conditions.

Fix leakage type \(\lambda_i\) and abbreviate \(Q_{is}:=Q_i(\theta_{is},\lambda_i;S^{\mathcal L})\), \(U_{is}:=U_i(\theta_{is},\lambda_i;S^{\mathcal L})\), and \(M_{is}:=M_i(\theta_{is},\lambda_i;S^{\mathcal L})\). Bidder \(i\) can always imitate the equilibrium strategy of any other payoff type while keeping her leakage type fixed. Hence, interim IR and the adjacent downward incentive constraints imply \(U_{i1}\geq0\) and
\begin{equation}
  U_{is}-U_{i,s-1}\geq(\theta_{is}-\theta_{i,s-1})Q_{i,s-1}\qquad\text{for }s=2,\dots,K_i.
  \label{eq:appendix-adjacent-downward-ic}
\end{equation}
Therefore \(U_{is}\geq R_i^Q(\theta_{is})\), where \(R_i^Q(\theta_{i1})=0\) and \(R_i^Q(\theta_{is})=\sum_{r=1}^{s-1}(\theta_{i,r+1}-\theta_{ir})Q_{ir}\) for \(s\geq2\). Since \(M_{is}=\theta_{is}Q_{is}-U_{is}\), summation by parts gives
\begin{align}
  \sum_{s=1}^{K_i}f_i(\theta_{is})M_{is}&\leq\sum_{s=1}^{K_i}f_i(\theta_{is})\bigl(\theta_{is}Q_{is}-R_i^Q(\theta_{is})\bigr)\notag\\
  &=\sum_{s=1}^{K_i}f_i(\theta_{is})\phi_i(\theta_{is})Q_{is}.
  \label{eq:conditional-virtual-surplus-bound}
\end{align}
Because every payoff type has positive probability, equality holds if and only if \(U_{i1}=0\) and every inequality in Equation~\eqref{eq:appendix-adjacent-downward-ic} binds. Apply Equation~\eqref{eq:conditional-virtual-surplus-bound} and then Equation~\eqref{eq:ironing-identity}, average over leakage types, and sum over bidders to obtain the first inequality in Equation \eqref{eq:auction-revenue-bound} and its equality conditions.

For every \((\theta,\omega)\), feasibility implies
\[
  \sum_{i\in N}\overline \phi_i(\theta_i)q_i^{S^{\mathcal L}}(\theta,\omega)\leq\max\left\{0,\max_{i\in N}\overline\phi_i(\theta_i)\right\}.
\]
Equality holds if and only if the object is allocated only to a bidder whose ironed virtual value is nonnegative and maximal among all bidders, while the seller retains the object only if every bidder's ironed virtual value is nonpositive. Taking expectations proves the second inequality in Equation \eqref{eq:auction-revenue-bound} and completes the proof. \qed

\bigskip

A BIC-implementable allocation rule \(q\) is revenue maximizing if and only if its interim allocation probabilities satisfy the ironing conditions and, for every \(\theta\in\Theta\), \(q_i(\theta)>0\) only if \(\overline\phi_i(\theta_i)\geq0\) and \(\overline\phi_i(\theta_i)=\max_{j\in N}\overline\phi_j(\theta_j)\), while \(q_0(\theta)>0\) only if \(\max_{j\in N}\overline\phi_j(\theta_j)\leq0\). Indeed, the canonical transfers attain equality in the first inequality of Lemma~\ref{lem:auction-revenue-bound} exactly when the ironing conditions hold, and the stated allocation condition is exactly the equality condition for the second inequality. This characterizes the revenue-maximizing allocation rule \(q^{\mathrm R}\) fixed in the main text.

Such allocation rules exist. Fix a priority order on \(N\cup\{0\}\). If \(\max_{j\in N}\overline\phi_j(\theta_j)>0\), let \(q^{\mathrm D}\) assign the object to the highest-priority bidder attaining this maximum. If \(\max_{j\in N}\overline\phi_j(\theta_j)\leq0\), let it select the highest-priority element among the seller and the bidders whose ironed virtual value is zero. The rule \(q^{\mathrm D}\) is BIC implementable because \(\overline\phi_i(\theta_i)\) is nondecreasing in \(\theta_i\), so raising \(\theta_i\) cannot cause bidder \(i\) to lose. If \(D_{is}>0\), then \(\overline\phi_i(\theta_{i,s-1})=\overline\phi_i(\theta_{is})\), and the fixed priority rule gives \(q_i^{\mathrm D}(\theta_{i,s-1},\theta_{-i})=q_i^{\mathrm D}(\theta_{is},\theta_{-i})\) for every \(\theta_{-i}\). Hence \(Q_i^{q^{\mathrm D}}\) satisfies the ironing conditions, so \(q^{\mathrm D}\) is revenue maximizing and \(\overline{\operatorname{Rev}}(q^{\mathrm D})=\operatorname{Rev}^{\mathrm M}\).

\subsection{A leakage-proof Dutch implementation}
\label{app:dutch}

Given the allocation rule \(q^{\mathrm D}\) defined above, for each bidder \(i\), write \(Q_i^{\mathrm D}:=Q_i^{q^{\mathrm D}}\) and \(R_i^{\mathrm D}:=R_i^{q^{\mathrm D}}\), and define the canonical expected payment by \(M_i^{\mathrm D}(\theta_i):=\theta_iQ_i^{\mathrm D}(\theta_i)-R_i^{\mathrm D}(\theta_i)\). Whenever \(Q_i^{\mathrm D}(\theta_i)>0\), define
\[
  \beta_i(\theta_i):=\frac{M_i^{\mathrm D}(\theta_i)}{Q_i^{\mathrm D}(\theta_i)}.
\]
If two payoff types of bidder \(i\) have the same ironed virtual value, the fixed-priority rule assigns them the same allocation probability, and the canonical payment formula then gives them the same expected payment. Hence, \(\beta_i(\theta_i)\) depends only on the ironed virtual value \(c=\overline\phi_i(\theta_i)\), which we use as the bidder's index; we therefore write it as \(\beta_i(c)\). The seller is treated as having index zero.

The \textbf{Dutch auction} orders all bidder-index pairs \((i,c)\) that rank above the seller from highest to lowest index, breaking ties according to the fixed priority order, and queries these pairs sequentially in that order. At the stage associated with \((i,c)\), bidder \(i\) publicly chooses whether to accept the object at payment \(\beta_i(c)\). Acceptance ends the auction, while the seller keeps the object if every offer is rejected. The prescribed strategy for bidder \(i\) of type \(\theta_i\) is to reject all earlier offers and accept when the pair corresponding to \(c=\overline\phi_i(\theta_i)\) is reached. Types with \(Q_i^{\mathrm D}(\theta_i)=0\) reject every offer.

To verify the prescribed strategy, observe that each rejection leads to a unique next offer. Hence, every pure deviation is characterized by the first offer the bidder accepts, or by rejecting all offers. If bidder \(i\)'s first accepted offer corresponds to \(c=\overline\phi_i(\widehat\theta_i)\), then, against the prescribed strategies of the other bidders, the offer is reached with probability \(Q_i^{\mathrm D}(\widehat\theta_i)\), and a type \(\theta_i\) obtains
\[
  Q_i^{\mathrm D}(\widehat\theta_i)\bigl(\theta_i-\beta_i(\overline\phi_i(\widehat\theta_i))\bigr)=\theta_iQ_i^{\mathrm D}(\widehat\theta_i)-M_i^{\mathrm D}(\widehat\theta_i).
\]
This is exactly her interim payoff from reporting \(\widehat\theta_i\) in the BIC direct mechanism with allocation rule \(q^{\mathrm D}\) and canonical transfers. BIC and interim IR therefore imply that the prescribed stopping strategy is optimal against every pure, and hence every mixed, deviation. The resulting equilibrium implements \(q^{\mathrm D}\), interim implements \(R^{q^{\mathrm D}}\), and yields \(\operatorname{Rev}^{\mathrm M}\).

The same comparison establishes sequential optimality at every decision history reached with positive probability: there is a unique public history leading to each offer, so the ex-ante payoff difference equals its positive reach probability times the corresponding continuation-payoff difference. Behavior at histories having zero probability for a given payoff type can be completed to obtain a PBE. Finally, only one bidder acts at each stage, so there are no concurrent actions from which information can leak. The auction is therefore leakage-proof.

\subsection{Proof of Proposition~\ref{prop:optimal-revenue-arbitrary-distributions}}

Lemma~\ref{lem:auction-revenue-bound} implies \(\operatorname{Rev}^\star(\mathcal L)\leq\operatorname{Rev}^{\mathrm M}\) for every finite leakage type space \(\mathcal L\). The Dutch auction associated with \(q^{\mathrm D}\) discussed in Appendix \ref{app:dutch} attains \(\operatorname{Rev}^{\mathrm M}\) in every leakage type space, so \(\operatorname{Rev}^\star(\mathcal L)=\operatorname{Rev}^{\mathrm M}\).

Now fix the allocation rule \(q^{\mathrm R}\) from the proposition. Because it is BIC implementable and its canonical transfers yield \(\overline{\operatorname{Rev}}(q^{\mathrm R})=\operatorname{Rev}^{\mathrm M}\), the necessity statement follows from the fixed-allocation theorem applied to \(q^{\mathrm R}\). Conversely, a leakage-proof default mechanism that satisfies \(\mathcal C_{q^{\mathrm R}}\) and interim implements \(R^{q^{\mathrm R}}\) attains \(\overline{\operatorname{Rev}}(q^{\mathrm R})=\operatorname{Rev}^{\mathrm M}\) in every leakage type space. It therefore maximizes revenue conditional on \(q^{\mathrm R}\) and, by the first part of the proposition, maximizes revenue without qualification.

\subsection{Order-by-Order Efficient and Optimal Auction}
\label{app:fixed-allocation-tie}

There are two bidders and one object. Values are independent and satisfy
\[
  \Pr(\theta_i=L)=\frac45,
  \qquad
  \Pr(\theta_i=H)=\frac15,
  \qquad
  L=1,
  \qquad
  H=5.
\]
The virtual values are \(\phi(L)=0\) and \(\phi(H)=5\), so the Myersonian optimal revenue is
\[
  \operatorname{Rev}^{\mathrm M}
  =
  5\Pr\bigl(\max_i\theta_i=H\bigr)
  =
  \frac95.
\]

At stage~1, the bidders play rock--paper--scissors to determine their priority in breaking ties. Each bidder chooses \(R\), \(P\), or \(S\). Let \(w_i\in\{0,1/2,1\}\) denote bidder \(i\)'s priority weight: \(w_i=1\) if she wins, \(w_i=0\) if she loses, and \(w_i=1/2\) if the game ends in a tie. At stage~2, the bidders simultaneously report \(L\) or \(H\). For the priority weights selected at stage~1, allocation and payments are
\[
  \begin{array}{c|c|c}
    \text{reports} & \text{allocation} & \text{payments}\\
    \hline
    (L,L) & q_i=w_i & m_i=q_i\\
    (H_i,L_j) & q_i=1,\ q_j=0 & m_i=5-4w_i,\ m_j=0\\
    (H,H) & q_i=\frac14+\frac12w_i & m_i=5q_i.
  \end{array}
\]
For every fixed priority pair, truthful reporting is weakly dominant.\footnote{This game resembles the auction of \citet{Harris-Raviv-AllocationMechanismsDesign-1981} with tie-breaking probabilities depending on the bid level (cf also \citet{Bulow-Roberts-SimpleEconomicsOptimal-1989}).
}
A high type obtains \(4w_i\) against report \(L\) and zero against report \(H\), independently of her own report. A low type obtains zero from reporting \(L\); reporting \(H\) gives \(-4(1-w_i)\leq0\) against \(L\) and \(-4(1/4+w_i/2)<0\) against \(H\).

In every finite leakage type space, consider the following strategy. A bidder who observes the opponent's stage-1 action chooses its unique beater; a bidder who observes no concurrent action mixes uniformly over \(R,P,S\); and both bidders report truthfully. A low type is indifferent over priority. A high type's stage-2 payoff is \(4\Pr(\theta_j=L)=16/5\) times her priority weight. An observer therefore chooses the beater. A nonobserver either faces an equally informed opponent, against whose uniform mixture every action gives expected priority \(1/2\), or a more informed opponent, who observes and beats her action, in which case every action gives priority zero. Thus the uniform mixture is optimal.

Fix a realized leakage order and write \(z:=\mathbb E[w_1]\). Priority is independent of values, and revenue at the four payoff-type profiles is
\[
  \begin{array}{c|cccc}
    (\theta_1,\theta_2) & (L,L) & (H,L) & (L,H) & (H,H)\\
    \hline
    \text{revenue} & 1 & 5-4z & 1+4z & 5.
  \end{array}
\]
Hence
\[
  \mathbb E[\text{revenue}]
  =
  \frac{16}{25}
  +
  \frac4{25}(5-4z)
  +
  \frac4{25}(1+4z)
  +
  \frac5{25}
  =
  \frac95.
\]
The allocation is efficient at every payoff-type profile and leakage order, and the revenue equals the Myersonian optimal revenue conditional on every order. However, the interim-payoff schedule depends on the realized leakage order.

Let
\[
  z_\omega
  :=
  \mathbb E[w_1\mid\omega]
\]
under the equilibrium above. Conditional on the realized leakage order \(\omega\), bidder~1's interim allocation probability when her value is \(L\) and her interim utility when her value is \(H\) are
\[
  Q_1(L\mid\omega)
  =
  \frac45 z_\omega,
  \qquad
  U_1(H\mid\omega)
  =
  \frac{16}{5}z_\omega.
\]
Similarly,
\[
  Q_2(L\mid\omega)
  =
  \frac45(1-z_\omega),
  \qquad
  U_2(H\mid\omega)
  =
  \frac{16}{5}(1-z_\omega).
\]
Under no leakage, both bidders mix uniformly and \(z_\omega=1/2\). If bidder~1 is more informed, she chooses the action that beats bidder~2's observed action and \(z_\omega=1\). If bidder~2 is more informed, \(z_\omega=0\). Thus the allocation rule and the interim-rent schedule vary with the leakage order, even though the allocation is efficient and expected revenue remains equal to \(9/5\) under every order.

The no-leakage equilibrium just described is not leakage-proof: after observing the opponent's stage-1 action, a high type can profitably replace the prescribed uniform mixture with the unique action that beats the observed action.

\subsection{Proof of Proposition \ref{prop:multiunit-beyond-dic}}\label{app:two-unit-leakage-proof}
Consider two bidders and two identical units. Each bidder's payoff type is independently and uniformly distributed on
\[
  \Theta_i
  =
  \left\{
    \theta^1=(1,1),\;
    \theta^2=(2,1),\;
    \theta^3=(5,3)
  \right\},
\]
where the coordinates are marginal values for the first and second unit. \cite{Gershkov-Goeree-Kushnir-Moldovanu-Shi-EquivalenceBayesianDominant-2013} show that in this environment
\[
  \operatorname{Rev}^{\mathrm{BIC}}
  =
  \frac{191}{45},
  \qquad
  \operatorname{Rev}^{\mathrm{DIC}}
  =
  \frac{38}{9}.
\]
Thus $\operatorname{Rev}^{\mathrm{BIC}} > \operatorname{Rev}^{\mathrm{DIC}}.$ It remains to construct a leakage-proof public sequential mechanism
that attains $191/45$.

The mechanism tests reports sequentially:
\[
  \begin{array}{ll}
    \text{bidder 1:}
    & \text{claim }\theta^3\text{ or pass},\\
    \text{bidder 2:}
    & \text{claim }\theta^3\text{ or pass},\\
    \text{bidder 1, if she passed }\theta^3:
    & \text{claim }\theta^2\text{ or pass},\\
    \text{bidder 2, if she passed }\theta^3:
    & \text{claim }\theta^2\text{ or pass}.
  \end{array}
\]
Claiming $\theta^3$ assigns report $3$, claiming $\theta^2$ after passing assigns report $2$, and passing both tests assigns report $1$. The allocation and payments will be chosen such that the prescribed strategy is truthful: type $\theta^3$ claims at the $\theta^3$ test, type $\theta^2$ passes the $\theta^3$ test and claims at the $\theta^2$ test, and type $\theta^1$ passes both tests.

Rows in the following table are bidder 1's report and columns are bidder 2's report. Each entry gives
\[
  \bigl(
    \Pr[\text{bidder 1 receives both units}],
    \Pr[\text{one unit each}],
    \Pr[\text{bidder 2 receives both units}]
  \bigr).
\]
The terminal allocation rule is
\[
  \begin{array}{c|ccc}
    &1&2&3\\ \hline
    1&
    (\frac12,0,\frac12)&
    (\frac12,0,\frac12)&
    (0,0,1)
    \\[1mm]
    2&
    (\frac25,0,\frac35)&
    (0,1,0)&
    (0,0,1)
    \\[1mm]
    3&
    (1,0,0)&
    (1,0,0)&
    (0,1,0).
  \end{array}
\]
The corresponding expected payments $(m_1,m_2)$ are
\[
  \begin{array}{c|ccc}
    &1&2&3\\ \hline
    1&
    (1,1)&
    (1,1)&
    (0,8)
    \\[1mm]
    2&
    (\frac15,\frac65)&
    (2,\frac75)&
    (0,\frac75)
    \\[1mm]
    3&
    (2,0)&
    (8,0)&
    (5,5).
  \end{array}
\]
The displayed expected payments can be implemented by payments contingent on the realization of the allocation lottery.

We verify sequential rationality by backward induction. First consider bidder 1's lower $\theta^2$ test after both bidders have passed their $\theta^3$ tests. Conditional on this history, bidder 2 is equally likely to have type $\theta^1$ or $\theta^2$. For bidder 1's three possible payoff types, reporting $1$ and reporting $2$ yield utilities of, respectively,
\[
  \left(0,\frac12,3\right)
  \qquad\text{and}\qquad
  \left(-\frac15,\frac12,3\right).
\]
Thus type $\theta^1$ strictly prefers to pass, type $\theta^2$ is willing to claim, and even a type $\theta^3$ that has deviated by passing the earlier test cannot profit at this history. If bidder 2 has instead claimed $\theta^3$, both lower reports give bidder 1 payoff zero, so the prescribed continuation remains optimal.

Next consider bidder 2's lower test. If bidder 1 has reported $1$,
reports $1$ and $2$ both yield utilities
\[
  \left(0,\frac12,3\right)
\]
to bidder 2's three possible payoff types. If bidder 1 has reported $2$, reports $1$ and $2$ yield, respectively,
\[
  \left(0,\frac35,\frac{18}{5}\right)
  \qquad\text{and}\qquad
  \left(-\frac25,\frac35,\frac{18}{5}\right).
\]
If bidder 1 has reported $3$, both lower reports give zero. Hence type $\theta^1$ weakly prefers to pass, type $\theta^2$ is willing to claim, and a type $\theta^3$ that has deviated by passing the upper test gains nothing from choosing a different lower report.

Consider now the upper $\theta^3$ tests. At bidder 1's initial test, claiming $\theta^3$ immediately gives expected payoffs
\[
  \left(
    -\frac{10}{3},
    -\frac{7}{3},
    2
  \right)
\]
to types $(\theta^1,\theta^2,\theta^3)$, whereas passing and then behaving optimally at the lower test gives
\[
  \left(
    0,
    \frac13,
    2
  \right).
\]
Thus the two lower types pass and the highest type is willing to claim.

At bidder 2's $\theta^3$ test, first suppose bidder 1 has already claimed $\theta^3$. Claiming and passing give, respectively,
\[
  (-4,-3,0)
  \qquad\text{and}\qquad
  (0,0,0).
\]
Hence the lower types pass and the highest type is willing to claim. If bidder 1 has instead passed $\theta^3$, bidder 2 knows that bidder 1 is equally likely to have type $\theta^1$ or $\theta^2$. Claiming and passing, followed by optimal behavior at the lower test, give
\[
  \left(
    -\frac{27}{10},
    -\frac{17}{10},
    \frac{33}{10}
  \right)
  \qquad\text{and}\qquad
  \left(
    0,
    \frac{11}{20},
    \frac{33}{10}
  \right),
\]
respectively. Again the two lower types pass and the highest type is willing to claim.

These comparisons establish sequential rationality at every decision history, including histories reached after a bidder has deviated from her prescribed earlier action. Since the game is finite and has perfect recall, the one-deviation principle completes the equilibrium verification. Beliefs at all relevant public histories are determined by Bayes' rule.

The truthful equilibrium is interim individually rational. Bidder 1's interim utilities for types $(\theta^1,\theta^2,\theta^3)$ are
\[
  \left(0,\frac13,2\right),
\]
while bidder 2's are
\[
  \left(0,\frac{11}{30},\frac{11}{5}\right).
\]
All are nonnegative.

Finally, summing total expected payments over the nine equally likelytype profiles gives
\[
  9\,\operatorname{Rev}
  =
  2+2+8
  +\frac75+\frac{17}{5}+\frac75
  +2+8+10
  =
  \frac{191}{5}.
\]
Therefore
\[
  \operatorname{Rev}
  =
  \frac{191}{45}
  =
  \operatorname{Rev}^{\mathrm{BIC}}.
\]

At every strategic stage exactly one bidder is active. Hence there are no concurrent bidder actions from which information can leak. The same strategy profile and beliefs therefore constitute a PBE under every finite leakage type space. The mechanism is leakage-proof and attains the BIC-optimal revenue, even though $\operatorname{Rev}^{\mathrm{DIC}} < \operatorname{Rev}^{\mathrm{BIC}}$.

\subsection{Proof of Proposition \ref{prop:multiunit-optimal}}
Fix a finite leakage type space $\mathcal L$ and an interim-IR PBE$(S^{\mathcal L},\mu^{\mathcal L})$ of an arbitrary finite multi-unit auction. For bidder $i$, write

$$\Delta\theta_i^s:=\theta_i^{s+1}-\theta_i^s,\qquad s=1,\ldots,K_i-1.$$

By coordinatewise ordering of physical types, $\Delta\theta_i^s\geq 0$ coordinatewise.

For each leakage type $\lambda_i$, let $Q_i^s(\lambda_i)$ denote the interim tail-allocation vector of bidder $i$ conditional on payoff type $\theta_i^s$ and leakage type $\lambda_i$, let $M_i^s(\lambda_i)$ denote her interim expected payment, and let

$$U_i^s(\lambda_i)=\theta_i^s\cdot Q_i^s(\lambda_i)-M_i^s(\lambda_i)$$

denote her interim utility.

A bidder of type $(\theta_i^{s+1},\lambda_i)$ can imitate the complete equilibrium strategy of type $(\theta_i^s,\lambda_i)$. Sequential rationality at the initial history therefore implies
\begin{equation}\label{eq:MU-downward-IC-proof}
  U_i^{s+1}(\lambda_i)-U_i^s(\lambda_i)
  \geq
  \Delta\theta_i^s\cdot Q_i^s(\lambda_i),
  \qquad
  s=1,\ldots,K_i-1.
\end{equation}
Interim IR implies $U_i^1(\lambda_i)\geq0$. Iterating
\eqref{eq:MU-downward-IC-proof} gives
\[
  U_i^s(\lambda_i)
  \geq
  U_i^1(\lambda_i)
  +
  \sum_{r=1}^{s-1}
  \Delta\theta_i^r\cdot Q_i^r(\lambda_i).
\]
Hence
\begin{align*}
  \sum_{s=1}^{K_i}f_i^s M_i^s(\lambda_i)
  &=
  \sum_{s=1}^{K_i}f_i^s
  \left[
    \theta_i^s\cdot Q_i^s(\lambda_i)
    -
    U_i^s(\lambda_i)
  \right]
  \\
  &\leq
  \sum_{s=1}^{K_i}f_i^s
  \theta_i^s\cdot Q_i^s(\lambda_i)
  -
  U_i^1(\lambda_i)
  -
  \sum_{s=1}^{K_i}f_i^s
  \sum_{r=1}^{s-1}
  \Delta\theta_i^r\cdot Q_i^r(\lambda_i)
  \\
  &=
  \sum_{s=1}^{K_i}f_i^s
  \theta_i^s\cdot Q_i^s(\lambda_i)
  -
  U_i^1(\lambda_i)
  -
  \sum_{r=1}^{K_i-1}
  (1-F_i^r)
  \Delta\theta_i^r\cdot Q_i^r(\lambda_i)
  \\
  &=
  \sum_{s=1}^{K_i}f_i^s
  \phi_i^s\cdot Q_i^s(\lambda_i)
  -
  U_i^1(\lambda_i)
  \\
  &\leq
  \sum_{s=1}^{K_i}f_i^s
  \phi_i^s\cdot Q_i^s(\lambda_i).
\end{align*}
Averaging over leakage types and summing over bidders therefore yields
\begin{equation}\label{eq:MU-revenue-bound-proof}
  \operatorname{Rev}(G,S^{\mathcal L})
  \leq
  \mathbf E_{f\otimes\nu}
  \left[
    \sum_{i\in N}
    \phi_i(\theta_i)\cdot
    q_i^{S^{\mathcal L}}(\theta,\omega)
  \right].
\end{equation}
For every realized $(\theta,\omega)$, feasibility implies
\[
  \sum_{i\in N}
  \phi_i(\theta_i)\cdot
  q_i^{S^{\mathcal L}}(\theta,\omega)
  \leq
  \max_{a\in Y}
  \sum_{i\in N}
  \phi_i(\theta_i)\cdot y(a_i).
\]
Thus,
\begin{equation}\label{eq:MU-final-bound-proof}
  \operatorname{Rev}(G,S^{\mathcal L})
  \leq
  \mathbf E_f
  \left[
    \max_{a \in Y}
    \sum_{i\in N}
    \phi_i(\theta_i)\cdot y(a_i)
  \right].
\end{equation}

It remains to show that the bound is attainable by a dominant-strategy mechanism.

Fix a strict total order on $Y$, independent of payoff types and leakage orders. For each type profile $\theta$, let $a^R(\theta)$ be the highest-priority allocation among those maximizing $\sum_{i\in N}\phi_i(\theta_i)\cdot y(a_i).$ Write
\[
  q_i^R(\theta):=y(a_i^R(\theta)).
\]
We first show that $q_i^R$ is coordinatewise non-decreasing in bidder $i$'s type.

Fix $i$, $\theta_{-i}$, and two adjacent types $\theta_i^s$ and
$\theta_i^{s+1}$. Let
\[
  a=a^R(\theta_i^s,\theta_{-i}),
  \qquad
  b=a^R(\theta_i^{s+1},\theta_{-i}).
\]
Suppose, toward a contradiction, that $b_i<a_i$. Optimality of $a$ at type $s$ and of $b$ at type $s+1$ gives
\begin{align*}
  \phi_i^s\cdot y(a_i)
  +
  \sum_{j\neq i}\phi_j(\theta_j)\cdot y(a_j)
  &\geq
  \phi_i^s\cdot y(b_i)
  +
  \sum_{j\neq i}\phi_j(\theta_j)\cdot y(b_j),
  \\
  \phi_i^{s+1}\cdot y(b_i)
  +
  \sum_{j\neq i}\phi_j(\theta_j)\cdot y(b_j)
  &\geq
  \phi_i^{s+1}\cdot y(a_i)
  +
  \sum_{j\neq i}\phi_j(\theta_j)\cdot y(a_j).
\end{align*}
Adding the two inequalities gives
\[
  \bigl(\phi_i^{s+1}-\phi_i^s\bigr)
  \cdot
  \bigl(y(b_i)-y(a_i)\bigr)
  \geq0.
\]
By assumption, $\phi_i^{s+1}-\phi_i^s\geq0$ coordinatewise, while $b_i<a_i$ implies $y(b_i)-y(a_i)\leq0$ coordinatewise. The reverse weak inequality therefore also holds, so equality must obtain. Consequently both of the preceding optimality inequalities bind: $a$ and $b$ are both virtual-surplus maximizers at both type profiles. But the same strict priority rule is used at both profiles. Selection of $a$ at the first profile requires $a$ to have higher priority than $b$, whereas selection of $b$ at the second requires the reverse. This is impossible. Hence
\[
  a_i^R(\theta_i^{s+1},\theta_{-i})
  \geq
  a_i^R(\theta_i^s,\theta_{-i}),
\]
and therefore
\begin{equation}\label{eq:MU-monotone-allocation-proof}
  q_i^R(\theta_i^{s+1},\theta_{-i})
  \geq
  q_i^R(\theta_i^s,\theta_{-i})
\end{equation}
coordinatewise.

We now construct payments. Fix $i$ and $\theta_{-i}$, and abbreviate
\[
  q_i^s:=q_i^R(\theta_i^s,\theta_{-i}).
\]
Define the ex-post rent of the lowest type by $R_i^1(\theta_{-i})=0$ and, for $s\geq2$, let
\[
  R_i^s(\theta_{-i})
  =
  \sum_{r=1}^{s-1}
  \Delta\theta_i^r\cdot q_i^r.
\]
Set
\begin{equation}\label{eq:MU-envelope-payment-proof}
  m_i^R(\theta_i^s,\theta_{-i})
  =
  \theta_i^s\cdot q_i^s
  -
  R_i^s(\theta_{-i}).
\end{equation}

These payments make truthful reporting a dominant strategy. To see this, fix the opponents' reports. If $r<s$, the utility gain of type $\theta_i^s$ from reporting truthfully rather than reporting $\theta_i^r$ is
\begin{align*}
  R_i^s
  -
  \left[
    R_i^r+
    (\theta_i^s-\theta_i^r)\cdot q_i^r
  \right]
  &=
  \sum_{k=r}^{s-1}
  \Delta\theta_i^k
  \cdot
  (q_i^k-q_i^r)
  \\
  &\geq0,
\end{align*}
where the inequality follows from \eqref{eq:MU-monotone-allocation-proof} and $\Delta\theta_i^k\geq0$ coordinatewise. If $r>s$, the corresponding gain is
\begin{align*}
  R_i^s
  -
  \left[
    R_i^r+
    (\theta_i^s-\theta_i^r)\cdot q_i^r
  \right]
  &=
  \sum_{k=s}^{r-1}
  \Delta\theta_i^k
  \cdot
  (q_i^r-q_i^k)
  \\
  &\geq0.
\end{align*}
Thus truthful reporting is optimal for every realization of opponents' reports. Moreover, $R_i^s(\theta_{-i})\geq0$, so the mechanism is ex-post, and hence interim, individually rational.

Finally, under truthful reporting, the lowest type receives zero rent and every adjacent downward constraint used in deriving \eqref{eq:MU-revenue-bound-proof} binds state by state:
\[
  R_i^{s+1}(\theta_{-i})-R_i^s(\theta_{-i})
  =
  \Delta\theta_i^s\cdot
  q_i^R(\theta_i^s,\theta_{-i}).
\]
Hence the first inequality in the revenue bound is an equality. Since $a^R$ maximizes virtual surplus pointwise, the second inequality is an equality as well. The mechanism therefore attains the right-hand side of \eqref{eq:MU-final-bound-proof}.

Because truthful reporting is a dominant strategy, observing any subset of opponents' concurrent reports cannot create a profitable deviation. The direct mechanism is therefore leakage-proof. Since the upper bound applies to every interim-IR PBE under every finite leakage type space, the claimed revenue is optimal under leakage.

\section{Additional Material for Section~\ref{sec:discussion}}

\subsection{Pure-Strategy EPIC and Leakage-Proofness}
\label{app:pure-epic-proof}

Fix a pure-strategy EPIC default mechanism \((G,S^G)\). At any full-observation history and any fixed payoff-type profile, the opponents' concurrent and future default actions are deterministic functions of their payoff types and the public history. Hence any continuation that conditions on those actions can be reproduced by a strategy in the game with information leakage against that fixed payoff-type profile. If such a continuation were strictly profitable, it would be strictly profitable under information leakage, contradicting leakage-proofness. The default mechanism therefore is leakage-proof.

\subsection{Proof of Theorem~\ref{thm:sequential-leakage}}
\label{app:proof-sequential-leakage}

Suppose first that \(S^G\) is a leakage-proof equilibrium in \(G\). Fix a player \(i\), and consider the degenerate leakage type space whose commonly known order makes \(i\) more informed than every other player and makes all other players equally informed. Player \(i\) then observes every opponent's concurrent action at every stage. Leakage-proofness supplies a PBE that agrees with \(S^G\) at every default-reachable private history. Because \(G\) is a default-pruned game and the payoff-type prior has full support, after any unilateral continuation by \(i\), every opponent who has followed the default remains at an own-compatible, and hence default-reachable, private history. The opponents therefore continue to follow \(S_{-i}^G\). Sequential rationality of the fully informed player at each full-observation history gives Definition~\ref{def:sequential-posterior-equilibrium}.

Conversely, suppose \(S^G\) is a sequential posterior equilibrium in \(G\), and fix an arbitrary finite leakage type space \(\mathcal L\). Extend \(S^G\) to leaked histories by ignoring the leaked component. Every default-reachable private history \(h_i^t=(h^t,a_O^t)\) is a coarsening of the full-observation histories \((h^t,a_{-i}^t)\) consistent with \(a_O^t\). Any continuation strategy available after \(h_i^t\) is also a full-observation continuation strategy that ignores the additional actions. The posterior at \(h_i^t\) is the corresponding mixture of the full-observation posteriors. Averaging the optimality inequalities in Definition~\ref{def:sequential-posterior-equilibrium} over the unobserved actions therefore shows that the leakage-blind default continuation is optimal at every default-reachable private history.

To complete a PBE, call a private history own-compatible if the player's own past actions agree with her default strategy. Replace moves at own-compatible histories by chance moves following \(S^G\), take a sequential equilibrium of the resulting finite auxiliary game, and use its strategies at own-incompatible histories. At own-compatible histories, optimality follows from the preceding paragraph; at own-incompatible histories, it follows from the auxiliary equilibrium. The resulting PBE agrees with \(S^G\) at every default-reachable private history. Since \(\mathcal L\) was arbitrary, \((G,S^G)\) is leakage-proof, and \(S^G\) is a leakage-proof equilibrium in \(G\).

\subsection{Common Knowledge of Leakage Orders}
\label{app:known-order-counterexample}

There are two players. Player \(i\)'s payoff type is independently and uniformly distributed on \(\Theta_i=\{\theta_i^1,\theta_i^2\}\). The target social choice function is
\[
  \chi(\theta_1^k,\theta_2^\ell)=x_{k\ell},
  \qquad
  k,\ell\in\{1,2\},
\]
where the four target outcomes are distinct. Utilities at these outcomes are
\[
  \begin{array}{c|rrrr}
    & x_{11} & x_{12} & x_{21} & x_{22}\\
    \hline
    u_1(\,\cdot\,,\theta_1^1) & 0 & 0 & 1 & -2\\
    u_1(\,\cdot\,,\theta_1^2) & -1 & -1 & 0 & 0
  \end{array}
  \qquad
  \begin{array}{c|rrrr}
    & x_{11} & x_{12} & x_{21} & x_{22}\\
    \hline
    u_2(\,\cdot\,,\theta_2^1) & 0 & -1 & 0 & -1\\
    u_2(\,\cdot\,,\theta_2^2) & -1 & 0 & -1 & 0.
  \end{array}
\]
Truthful reporting gives utility zero. Type \(\theta_1^1\)'s expected payoff from reporting \(2\) is \((1-2)/2=-1/2\); type \(\theta_1^2\) loses \(1\) from reporting \(1\); and either type of player~2 loses \(1\) from misreporting.

At stage~1, player~1 chooses from \(\{C,S_1,S_2\}\) and player~2 chooses from \(\{C,S\}\). If the action profile is \((C,C)\), the players concurrently report \(r,s\in\{1,2\}\) at stage~2 and the outcome is \(x_{rs}\). If the action profile is \((S_r,S)\), only player~2 reports \(s\) at stage~2 and the outcome is \(x_{rs}\). The mismatches \((C,S)\) and \((S_r,C)\) terminate at outcomes \(x^{CS}\) and \(x_r^{SC}\), respectively, with type-independent utilities
\[
  \bigl(u_1(x^{CS}),u_2(x^{CS})\bigr)=(-3,-1),
  \qquad
  \bigl(u_1(x_r^{SC}),u_2(x_r^{SC})\bigr)=(-3,1).
\]

Each commonly known deterministic order admits an equilibrium implementing \(\chi\). With no leakage, both players choose \(C\) and report truthfully. If player~2 is more informed, player~1 chooses \(C\), player~2 chooses \(C\) after every observed stage-1 action, and both report truthfully after \((C,C)\); player~2's report remains strictly truthful after observing player~1's report. If player~1 is more informed, player~2 chooses \(S\), type \(\theta_1^k\) responds with \(S_k\), and player~2 reports truthfully. After the off-path action \(C\) by player~2, player~1 chooses \(C\); at the subsequent reporting stage, player~2 reports truthfully and player~1 best responds to the observed report. These strategies are sequentially rational by the displayed utilities.

Nevertheless, every no-leakage equilibrium implementing \(\chi\) uses the concurrent branch. Exact implementation and full support rule out stage-1 mismatches, so all types must use either \((C,C)\) or the sequential convention \((S_k,S)\). The latter is not an equilibrium because player~2 can deviate from \(S\) to \(C\) and obtain utility \(1\) at \(x_k^{SC}\), rather than zero. Thus all types choose \(C\), and distinctness of the target outcomes forces truthful stage-2 reports.

Now let player~1 be more informed at stage~2. At the positive-probability payoff-type profile \((\theta_1^1,\theta_2^1)\), player~1 observes player~2's report \(s=1\). The default report \(r=1\) gives utility zero, whereas \(r=2\) gives utility \(1\) at \(x_{21}\). Support pruning removes the sequential branch but retains both stage-2 reports, because each report is used by one payoff type. The profitable deviation therefore survives pruning.

\end{document}